\documentclass{fundam}

\usepackage[utf8]{inputenc}

\usepackage[noend,linesnumbered,ruled,vlined]{algorithm2e}
\usepackage{amsmath,amssymb}
\usepackage{enumitem}
\usepackage{graphicx}
\usepackage{mathtools}
 \usepackage[super]{nth}
 \usepackage{thmtools}
 \usepackage{todonotes}

\usepackage{hyperref}
\usepackage[capitalise,nameinlink]{cleveref}
\crefname{enumi}{item}{items}
\crefname{figure}{Figure}{Figures}

\DeclareMathOperator{\al}{alph}
\DeclareMathOperator{\ar}{ar}
\DeclareMathOperator{\dist}{dist}
\DeclareMathOperator{\firstPosArch}{firstInArch}

\DeclareMathOperator{\isSubseq}{isSubseq}
\DeclareMathOperator{\lastpos}{last}
\DeclareMathOperator{\lastPosArch}{lastInArch}
\DeclareMathOperator{\Leq}{Leq}

\DeclareMathOperator{\Lex}{lexSmall}
\DeclareMathOperator{\mas}{MAS}
\DeclareMathOperator{\minArch}{minArch}
\DeclareMathOperator{\nextArray}{sameNext}
\DeclareMathOperator{\prevArray}{samePrev}
\DeclareMathOperator{\nextpos}{next}
\DeclareMathOperator{\rest}{r}
\DeclareMathOperator{\sas}{SAS}
\DeclareMathOperator{\sasRange}{sasRange}
\DeclareMathOperator{\sortedLast}{sortedLast}

\DeclareMathOperator{\masExtend}{masExt}

\DeclareMathOperator{\depth}{depth}
\DeclareMathOperator{\LA}{LA}
\DeclareMathOperator{\levelAncestor}{levelAncestor}
\DeclareMathOperator{\startSAS}{startSAS}
\DeclareMathOperator{\RMQ}{RMQ}
\DeclareMathOperator{\argmax}{argmax}
\DeclareMathOperator{\llo}{llo}

  \def\nth#1{#1$^{\text{th}}$}
  \newcommand{\len}[1]{|#1|}

\publyear{22}
\papernumber{2159}
\volume{189}
\issue{3-4}

 \finalVersionForARXIV

\begin{document}

\title{Absent Subsequences in Words}

\author{Maria Kosche,  Tore Koß\thanks{Address for correspondence: Institute for Computer Science,
	                      Georg-August University G\"ottingen, Germany},  Florin Manea, Stefan Siemer
    \\
	Institute for Computer Science\\
	Georg-August University G\"ottingen, Germany\\
    \{maria.kosche,tore.koss,florin.manea,stefan.siemer\}@cs.uni-goettingen.de
    }
	
\maketitle

\runninghead{M. Kosche et al.}{Absent Subsequences in Words}

\begin{abstract}
An absent factor of a string $w$ is a string $u$ which does not occur as a contiguous substring (a.k.a. factor) inside $w$. We extend this well-studied notion and define absent subsequences: a string $u$ is an absent subsequence of a string $w$ if $u$ does not occur as subsequence (a.k.a. scattered factor) inside $w$. Of particular interest to us are minimal absent subsequences, i.e., absent subsequences whose every subsequence is not absent, and shortest absent subsequences, i.e., absent subsequences of minimal length. We show a series of combinatorial and algorithmic results regarding these two notions. For instance: we give combinatorial characterisations of the sets of minimal and, respectively, shortest absent subsequences in a word, as well as compact representations of these sets; we show how we can test efficiently if a string is a shortest or minimal absent subsequence in a word, and we give efficient algorithms computing the lexicographically smallest absent subsequence of each kind; also, we show how a data structure for answering shortest absent subsequence-queries for the factors of a given string can be efficiently computed.
\end{abstract}


\section{Introduction}
A word $u$ is a subsequence (also called scattered factor or subword) of a string $w$ if there exist (possibly empty) strings $v_1, \ldots, v_{\ell+1}$ and $u_1, \ldots, u_\ell$ such that $u = u_1 \ldots u_\ell$ and $w = v_1 u_1 \ldots v_\ell u_\ell v_{\ell+1}$. In other words, $u$ can be obtained from $w$ by removing some of its letters.

The study of the relationship between words and their subsequences has been a central topic in combinatorics on words and string algorithms, as well as in language and automata theory (see, e.g.,
the chapter {\em Subwords} in \cite[Chapter 6]{Loth97} for an overview of the fundamental aspects of this topic). Subsequences appear in many areas of theoretical computer science, such as logic of automata theory~\cite{HalfonSZ17,KarandikarKS15,CSLKarandikarS,journals/lmcs/KarandikarS19,Kuske20,KuskeZ19,simonPhD,Simon72,Zetzsche16}, combinatorics on words~\cite{FreydenbergerGK15,Rigo19,LeroyRS17a,Mat04,RigoS15,Salomaa05,Seki12}, as well as algorithms \cite{DBLP:journals/tcs/Baeza-Yates91,DBLP:conf/fsttcs/BringmannC18,BringmannK18,Maier:1978,Wagner:1974}. From a practical point of view, subsequences are generally used to model corrupted or lossy representations of an original string, and appear, for instance, in applications related to formal verification, see \cite{HalfonSZ17,Zetzsche16} and the references therein, or in bioinformatics-related problems, see~\cite{sankoff}. \looseness=-1

In most investigations related to subsequences, comparing the sets of subsequences of two different strings is usually a central task. In particular, Imre Simon defined and studied (see \cite{Loth97,simonPhD,Simon72}) the relation $\sim_k$ (now called the Simon's Congruence) between strings having exactly the same set of subsequences of length at most $k$ (see, e.g., \cite{journals/lmcs/KarandikarS19} as well as the surveys \cite{Pin2004,Pin2019} and the references therein for the theory developed around $\sim_k$ and its applications). In particular, $\sim_k$ is a well-studied relation in the area of string algorithms. The problems of deciding whether two given strings are $\sim_k$-equivalent, for a given $k$, and to find the largest $k$ such that two given strings are $\sim_k$-equivalent (and their applications) were heavily investigated in the literature, see, e.g., \cite{DBLP:journals/jda/CrochemoreMT03,KufMFCS,Hebrard1991,garelCPM,SimonWords,DBLP:conf/wia/Tronicek02} and the references therein. Last year, optimal solutions were given for both these problems \cite{Barker2020,mfcs2020}. Two concepts seemed to play an important role in all these investigations: on the one hand, the notion of distinguishing word, i.e., the shortest subsequence present in one string and {\em absent} from the other. On the other hand, the notion of universality index of a string \cite{Barker2020,DayFKKMS21}, i.e., the largest $k$ such that the string contains as subsequences all possible strings of length at most $k$; that is, the length of the shortest subsequence {\em absent} from that string, minus $1$. \looseness=-1

Motivated by these two concepts and the role they play, we study in this paper the set of {\em absent subsequences} of a string $w$, i.e., the set of strings which {\em are not} subsequences of $w$. As such, our investigation is also strongly related to the study of {\em missing factors} (or missing words, MAWs) in strings, where the focus is on the set of strings which are not substrings (or factors) of $w$. The literature on the respective topic ranges from many very practical applications of this concept \cite{Barton2014,Chairungsee2012,Charalampopoulos2018,Crochemore2000,Pratas2020,Silva2015}
to deep theoretical results of combinatorial \cite{SolonCPM1,Crochemore1998,Fici2019,Fici2006,Fici2019a,Nakashima2020,Mignosi2002} or algorithmic nature \cite{Ayad2019,SolonCPM2,Barton2014,Barton2016,Charalampopoulos2018a,Crochemore2020,Fujishige2016}.
Absent subsequences are also related to the well-studied notion of patterns avoided by permutations, see for instance \cite{Kitaev11}, with the main difference being that a permutation is essentially a word whose letters are pairwise distinct.
\looseness=-1

Moreover, absent subsequences of a string (denoted by $w$ in the following) seem to naturally occur in many practical scenarios, potentially relevant in the context of reachability and avoidability problems. \looseness=-1

On the one hand, assume that $w$ is some string (or stream) we observe, which may represent e.g. the trace of some computation or, in a totally different framework, the DNA-sequence describing some gene. In this framework, absent subsequences correspond to sequences of letters avoided by the string $w$. As such, they can be an avoided sequence of events in the observed trace or an avoided scattered sequence of nucleotides in the given gene. Understanding the set of absent subsequences of the respective string, and in particular its extremal elements with respect to some ordering, as well as being able to quickly retrieve its elements and process them efficiently seems useful to us. \looseness=-1

\eject
On the other hand, when considering problems whose input is a set of strings, one could be interested in the case when the respective input can be compactly represented as the set of absent subsequences (potentially with some additional combinatorial properties, which make this set finite) of a given string. Clearly, one would then be interested in processing the given string and representing its set of absent subsequences by some compact data structure which further would allow querying it efficiently.\looseness=-1

In this context, our paper is focused on two particular classes of absent subsequences: {\em minimal absent subsequences} ($\mas$ for short), i.e., absent subsequences whose every subsequence is not absent, and {\em shortest absent subsequences} ($\sas$ for short), i.e., absent subsequences of minimal length. In \cref{combinatorics}, we show a series of novel combinatorial results: we give precise characterizations of the set of  minimal absent subsequences and shortest absent subsequences occurring in a word, as well as examples of words $w$ having an exponential number (w.r.t. the length of $w$) of  minimal absent subsequences and shortest absent subsequences, respectively. We also identify, for a given number $k$, a class of words having a maximal number of $\sas$, among all words whose $\sas$ have length $k$. \looseness=-1

We continue with a series of algorithmic results in \cref{sec:algorithms}. We first show a series of simple algorithms, useful to test efficiently if a string is a shortest or minimal absent subsequence in a word. Motivated by the existence of  words with exponentially large sets of  minimal absent subsequences and shortest absent subsequences, our main contributions show, in \cref{sasRepresentation,masRepresentation}, how to construct compact representations of these sets. These representations are fundamental to obtaining efficient algorithms querying the set of $\sas$ and $\mas$ of a word, and searching for such absent subsequences with certain properties or efficiently enumerating them. These results are based on the combinatorial characterizations of the respective sets combined with an involved machinery of data structures, which we introduce gradually for the sake of readability.
In \cref{sasRepresentation}, we show another main result of our paper,
where, for a given word $w$,
we construct in linear time a data structure for efficiently answering queries asking for the shortest absent subsequences in the factors of $w$ (note that the same problem was recently approached in the case of missing factors \cite{SolonCPM2}). \looseness=-1

The techniques used to obtain these results are a combination of combinatorics on words results with efficient data structures and algorithmic techniques. 

\section{Basic definitions}

Let $\mathbb{N}$ be the set of natural numbers, including $0$.
For $m, n \in \mathbb{N}$,
we let $[m:n] = \{m, m+1, \ldots, n\}$.
An alphabet $\Sigma$ is a nonempty finite set of symbols called {\em letters}.
A {\em string (also called word)} is a finite sequence of letters from $\Sigma$,
thus an element of the free monoid $\Sigma^\ast$.
For the rest of the paper,
we assume that the strings we work with are over an alphabet $\Sigma=\{1,2,\ldots,\sigma\}$. \looseness=-1

Let $\Sigma^+ = \Sigma^\ast \setminus \{\varepsilon\}$,
where $\varepsilon$ is the empty string.
The {\em length} of a string $w \in \Sigma^\ast$ is denoted by $\len w$.
The \nth{$i$} letter of $w \in \Sigma^\ast$ is denoted by  $w[i]$,
for $i \in [1:\len w]$.
For $m, n \in \mathbb{N}$,
we let $w[m:n] = w[m] w[m+1] \ldots w[n]$,
$\len w_a = |\{i \in [1:\len w] \mid w[i] = a \}|$. A string $u=w[m:n]$ is a {\em factor} of $w$, and we have $w = xuy$ for some $x,y \in \Sigma^\ast$.
If $x = \varepsilon$ (resp. $y = \varepsilon$),
$u$ is called a  {\em prefix} (resp. {\em suffix}) of $w$.
Let $\al(w) = \{x \in \Sigma \mid \len w_x > 0 \}$ be the smallest subset $S \subset \Sigma$ such that $w \in S^\ast$.
We can now introduce the notion of subsequence.

\begin{definition}
	We call $u$ a subsequence of length $k$ of $w$,
	where $\len w = n$,
	if there exist positions $1 \leq i_1 < i_2 < \ldots < i_k \leq n$,
	such that $u = w[i_1] w[i_2] \cdots w[i_k]$.
\end{definition}


We recall the notion of $k$-universality of a string as presented in \cite{Barker2020}.

\begin{definition}
	\label{def:k-universal}
		We call a word $w$ $k$-universal if any string $v$ of length $\len v \le k$  over $\al(w)$ appears as a subsequence of $w$.
For a word $w $,
		we define its universality index $\iota(w)$ to be the largest integer $k$
		such that $w$ is $k$-universal.
\end{definition}

If $\iota(w) = k$,
then $w$ is $\ell$-universal for all $\ell \leq k$.
Note that the universality index of a word $w$ is always defined w.r.t. the alphabet of the word $w$. For instance, $w = 01210$ is $1$-universal (as it contains all words of length $1$ over $\{0,1,2\}$
but would not be $1$-universal if we consider an extended alphabet $\{0,1,2,3\}$. The fact that the universality index is computed w.r.t. the alphabet of $w$ also means that every word is at least $1$-universal. Note that in our results we either investigate the properties of a given word $w$ or we show algorithms working on some input word $w$. In this context, the universality of the factors of $w$ and other words we construct is defined w.r.t. $\al(w)$. See \cite{Barker2020,DayFKKMS21} for a detailed discussion on this.

We recall the arch factorisation, introduced by Hebrard \cite{Hebrard1991}.

\begin{definition}\label{archfact}
	For $w \in \Sigma^\ast$, with $\Sigma=\al(w)$,
	the {\em arch factorisation} of $w$ is defined as $w = \ar_w(1) \cdots \ar_w(\iota(w)) \rest(w)$
	where for all $i \in [1:\iota(w)]$ the last letter of $\ar_{w}(i)$ occurs exactly once in $\ar_w(i)$,
	each arch $\ar_w(i)$ is $1$-universal,
	and $alph(\rest(w)) \subsetneq \Sigma$.
	The words $\ar_w(i)$ are called {\em arches} of $w$,
	$\rest(w)$ is called the {\em rest}. \looseness=-1
	
	Let $m(w) =\ar_w(1)[|\ar_w(1)|] \cdots \ar_w(k)[|\ar_w(k)|]$ be the word containing the unique last letters of each arch.	
\end{definition}

Note that every word has a unique arch factorisation and by definition each arch $\ar_w(i)$ from a word $w$ is $1$-universal.
By an abuse of notation,
we can write $i \in \ar_w(\ell)$ if $i$ is a natural number
such that $|\ar_w(1) \cdots \ar_w{(\ell-1)}| < i \leq |\ar_w(1) \cdots \ar_w{(\ell)}|$,
i.e., $i$ is a position of $w$ contained in the \nth{$\ell$} arch of $w$.

The main concepts discussed in this paper are the following.
\begin{definition} A word $v$ is an absent subsequence of $w$ if $v$ is not a subsequence of $w$.
An absent subsequence $v$ of $w$ is a minimal absent subsequence (for short, $\mas$) of $w$ if every proper subsequence of $v$ is a subsequence of $w$.
We will denote the set of all $\mas$ of $w$ by $\mas(w)$.
An absent subsequence $v$ of $w$ is a shortest absent subsequence (for short, $\sas$) of $w$ if $|v|\le |v'|$ for any other absent subsequence $v'$ of $w$.
We will denote the set of all $\sas$ of $w$ by $\sas(w)$.
\end{definition}

Note that any $\sas$ of $w$ has length $\iota(w)+1$ and $v$ is an $\mas$ of $w$ if and only if $v$ is absent and every subsequence of $v$ of length $|v|-1$ is a subsequence of $w$.

\section{Combinatorial properties of SAS and MAS}\label{combinatorics}

We begin with a presentation of several combinatorial properties of the $\mas$ and $\sas$.
Let us first take a closer look at $\mas$.

If $v=v[1]\cdots v[m+1]$ is an $\mas$ of $w$ then $v[1]\cdots v[m]$ is a subsequence of $w$. Hence, we can go left-to-right through $w$ and greedily choose positions $1 \le i_1 < \ldots < i_m\le n = |w|$ such that $v[1]\cdots v[m] = w[i_1]\cdots w[i_m]$ and $i_\ell$ is the leftmost occurrence of $w[i_\ell]$ in $w[i_{\ell-1}+1:n]$ (as described in \cref{alg:isSubseq} in \cref{sec:algorithms}).
Because $v$ itself is absent, $v[m+1]$ cannot occur in the suffix of $w$ starting at $i_{m}+1$. Furthermore, we know that $v[1]\cdots v[m-1]v[m+1]$ is a subsequence of $w$. Hence, $v[m+1]$ occurs in the suffix of $w$ starting at $i_{m-1}+1$. We deduce $v[m+1]\in \al(w[i_{m-1}+1:i_{m}])\setminus \al(w[i_{m}+1:n])$.
This argument is illustrated in \cref{fig:mas-intervals}
and can be applied inductively to deduce $v[k]\in \al(w[i_{k-2}+1:i_{k-1}])\setminus \al(w[i_{k-1}+1:i_{k}-1])$ for all $k\neq 1$. The choice of $v[1]\in \al(w)$ is arbitrary. More details are given in the proof to the following theorem.
For notational reasons we introduce $i_0=0$ and $i_{m+1}=n+1$.

\begin{figure}
	\centering
	\includegraphics{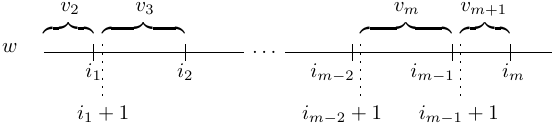}
	\caption{Illustration of positions and intervals inside word $w$}
	\label{fig:mas-intervals}\vspace*{-2mm}
\end{figure}

\begin{theorem}\label[theorem]{thm:mas}
Let $v,w\in \Sigma^\ast,~|v|=m+1$ and $|w|=n$, then $v$ is an $\mas$ of $w$ if and only if there are positions $0=i_0<i_1<\ldots <i_m <i_{m+1}= n+1$ such that all of the following conditions are satisfied.
\begin{enumerate}[label=(\roman*)]
\itemsep=0.9pt
	\item $v=w[i_1]\cdots w[i_m]v[m+1]$\label{thm:mas:abs}
	\item $v[1]\notin \al(w[1:i_1-1])$\label{thm:mas:v1}
	\item $v[k]\notin\al(w[i_{k-1}+1:i_k-1])$ for all $k\in[2: m+1]$\label{thm:mas:notin}	
	\item $v[k]\in \al(w[i_{k-2}+1:i_{k-1}])$ for all $k\in[2: m+1]$\label{thm:mas:in}
\end{enumerate}
\end{theorem}

\begin{proof}
{\crefname{enumi}{condition}{conditions}
Let $v$ be an $\mas$ of $w$ of length $|v|=m+1$.
By definition,
$v[1:m]$ is a subsequence of $w$,
hence we can choose positions $i_1,\ldots i_m$
such that $v[\ell]=w[i_\ell]$, for all $\ell \in [1:m]$.
We do this greedily,
that is,
we choose $i_1$ to be the leftmost occurrence of $v[1]$ in $w$,
$i_2$ be the leftmost occurrence of $v[2]$ in $w[i_1+1:n]$ and so on.
In the end,
$v[m+1]$ cannot occur in $w[i_m+1:n]$,
because $v$ is no subsequence of $w$.
Now, the positions $i_1, \ldots, i_m$ and $i_{m+1}=n+1$ clearly satisfy \cref{thm:mas:abs,thm:mas:v1,thm:mas:notin}.

\medskip
We show first that $k=m+1$ satisfies \cref{thm:mas:in}:
$v[1:m-1]v[m+1]$ is a subsequence of $w$,
hence $v[m+1]$ occurs in $w[i_{m-1}+1:n]$.
By \cref{thm:mas:notin} $v[m+1]$ does not occur in $w[i_m+1:n]$ hence
\cref{thm:mas:in} is true for $k=m+1$.
Now let $2 \leq \ell < m+1$, we make two observations:
\begin{enumerate}
\itemsep=0.9pt
	\item $v[\ell:m+1]$ is absent in $w[i_{\ell-1}+1:n]$
	because $v$ is absent in $w$,
	\label{one}
	\item $v[\ell+1:m+1]$ is a subsequence of $w[i_{\ell-1}+1:n]$
	because $v$ is a $\mas$ of $w$.
	\label{thr}
\end{enumerate}}
Combining both observations yields that $v[\ell:m+1]$ is a subsequence of $w[i_{\ell-2}+1:n]$ if and only if $v[\ell]$ occurs in $w[i_{\ell-2}+1:i_{\ell-1}]$. Since $v$ is an $\mas$ of $w,~v[1:\ell-2]v[\ell:m+1]$ is a subsequence of $w$, hence $v[\ell:m+1]$ is a subsequence of $w[i_{\ell-2}+1:n]$ and so $v[\ell]\in\al(w[i_{\ell-2}+1:i_{\ell-1}])$.

\medskip
{\crefname{enumi}{condition}{conditions}
Now let $i_0=0,~i_{m+1}=n+1$ and $i_1,\ldots, i_{m}$ be positions of $w$ satisfying \crefrange{thm:mas:abs}{thm:mas:in}.
From \cref{thm:mas:abs},
we have that $v = w[i_1]\cdots w[i_m]v[m+1]$.
We claim that $v$ is an $\mas$ of $w$.
Firstly, we divide $w$ into $w=w_1\cdots w_m \rest(w)$ where $w_k = w[i_{k-1}+1:i_k]$ and $\rest(w) = w[i_m+1:n]$.
By \cref{thm:mas:notin},
$v[m+1]$ doesn't occur in $\rest(w)$,
hence $v$ is absent in $w$.

\medskip
For the minimality,
we notice that by \cref{thm:mas:in},
$v[k]$ occurs in $w_{k-1}$.
Therefore, if we choose an arbitrary $\ell\le m+1$ and delete $v[\ell]$,
we have $v[1:\ell-1]=w[i_1]\cdots w[i_{\ell-1}]$ is a subsequence of $w_1\cdots w_{\ell-1}$,
and $v[\ell+1:m+1]$ is a subsequence of $w_\ell\cdots w_m$,
hence $v[1:\ell-1]v[\ell+1:m+1]$ is a subsequence of $w$.
The conclusion follows.}
\end{proof}

{\Crefname{enumi}{Property}{Properties}
\Cref{thm:mas:abs,thm:mas:v1,thm:mas:notin} (the latter for $k\le m$ only) are satisfied if we choose the positions $i_1,\ldots,i_m$ greedily, as described in the beginning of this section.}

\medskip
By \cref{thm:mas}, we have no restriction on the first letter of an $\mas$ and indeed we can find an $\mas$ starting with an arbitrary letter.

\begin{remark}\label[remark]{xton}
For every $x\in\al(w),~x^{|w|_x+1}$ is an $\mas$ of $w$, hence, for every choice of $x\in\al(w)$, we can find an $\mas$ $v$ starting with $x$.
\end{remark}

Using \cref{thm:mas} we can now determine the whole set of $\mas$ of a word $w$. This will be formalized later in \cref{thm:masRep2}. For now, we just give an example.

\begin{example}\label[example]{ex:mas}
	Let $w = 0011 \in \{0,1\}^\ast$ and we want to construct $v$, an $\mas$ of $w$.
	We start by choosing $v[1] = 0$.
	Then $i_1 = 1$ and by \cref{thm:mas:in} $v[2] \in \al(w[1:1]) = \{0\}$, so $i_2 = 2$ (by \cref{thm:mas:notin}).
	Again by \cref{thm:mas:in},
	we have $v[3] \in \al(w[2:2]) = \{0\}$.
	The letter $v[3]$ does not occur in $w[i_2+1:n]$,
	hence, $i_3 = n+1$, and $v = 0^3$ is an $\mas$ of $w$.
	If we let $v[1] = 1$,
	we have $i_1=3$.
	By \cref{thm:mas:in}, we have $v[2] \in \al(w[1:3]) = \{0,1\}$.
	If we choose $v[2] = 1$, we obtain $v = 1^3$ with an argument analogous to the first case.
	So let us choose $v[2] = 0$.
	Then $i_2 = n+1$, and $v=10$ is an $\mas$ of $w$.
	\Cref{thm:mas} claims $\mas(w)=\{0^3,10,1^3\}$
	and indeed $10$ is the only absent sequence of length $2$,
	and every word of length $\ge 3$ is either not absent ($001$, $011$ and $0011$)
	or contains $10,~0^3$ or $1^3$ as a subsequence.
\end{example}

From this example also follows that not every $\mas$ is an $\sas$. The converse is necessarily true. So, for any $\sas~v$ of $w$ we have $|v|=\iota(w)+1$, and we can find positions $1\le i_1<i_2<\ldots <i_{\iota(w)}\le n$ satisfying \cref{thm:mas}. The following theorem claims that every arch of $w$ (see \cref{archfact}) contains exactly one of these positions.

\begin{theorem}\label[theorem]{lem:sas}
Let $w=\ar_w(1)\cdots \ar_w(\iota(w))\rest(w)$ as in \cref{archfact}. Then, $v$ is an $\sas$ of $w$ if and only if there are positions $i_0=0,$ $i_\ell\in\ar_w(\ell)$ for all $1\le\ell\le\iota(w)$, and $i_{{\iota(w)}+1} = n+1$ satisfying \cref{thm:mas}.
\end{theorem}

\begin{proof}
Let $v$ be an $\sas$ of $w$. Every $\sas$ is an $\mas$, so we can choose positions $i_1,\ldots i_{\iota(w)}$ as in \cref{thm:mas}. We claim that every arch of $w$ contains exactly one of these positions. Because the number of these positions equals the number of arches, it suffices to show that every arch contains at least one $i_\ell$. Assume there is an arch $\ar_w(k)$ not containing any of the $i$s. Because,{\crefname{enumi}{condition}{conditions}
 by \cref{thm:mas:v1,thm:mas:notin}} of \cref{thm:mas}, we choose positions greedily and every arch is $1$-universal this is only possible if $v$ is a subsequence of $\ar_w(1)\cdots \ar_w(k-1)$. This is a contradiction because $v$ is absent from $w$.

\medskip
For the converse implication, we have a word $v=w[i_1]\cdots w[i_{\iota(w)}]v[\iota(w)+1]$ which satisfies \cref{thm:mas}. Hence $v$ is an $\mas$ of length $|v|=\iota(w)+1$ and thus an $\sas$.
\end{proof}

A way to efficiently enumerate all $\sas$ in a word will be given later in \cref{thm:sasRep}. Here, we only give an example based on a less efficient, but more intuitive, strategy of identifying the $\sas$ of a word.

\begin{example}
Let $w=012121012$ 
with $\iota(w)=2$, and the arch factorisation of $w$ is $w=012\cdot 1210\cdot 12$.
We construct $v$, an $\sas$ of $w$. By \cref{lem:sas}, we have $|v|=\iota(w)+1=3$ and by {\crefname{enumi}{condition}{conditions}\cref{thm:mas:in}} of \cref{thm:mas} the letter $v[3]$ does not occur in $\al(w[i_2:n])\supset \al(\rest(w))$, so $v[3]$ is not contained in $\al(r(w))$. Hence, $v[3]=0$ and its rightmost position in $\ar_w(2)$ is on position $7$. Therefore, $v[2]$ should not appear before position $7$ in $\ar_w(2)$ (as $v[1]$ appears in $\ar_w(1)$ for sure). So, $v[2] \notin\al(w[4:6])=\{1,2\}$ and $v[2]=0$. Ultimately, the rightmost occurrence of $v[2]$ in $\ar_w(1)$ is on position $1$, and we can  arbitrarily choose $v[1]\in\al(\ar_w(1))=\{0,1,2\}$.
We conclude that $\sas(w) = \{000, 100, 200\}$.
\end{example}

To better understand the properties of $\sas$ and $\mas$, we analyse some particular words. For the rest of this section, assume that $\sigma\geq 4$ is a large even natural constant.

\medskip
Firstly we let $v_1=1\cdot 2\cdots\sigma$ and $v_2=\sigma\cdot\sigma -1\cdots 1$. Then, for $k\in \mathbb N$, we define the words $A_{2k}= (v_1\cdot v_2)^{k}$, $A_{2k+1}= A_{2k}\cdot v_1$, and $B_k=v_1^k$. Each of these words has universality index $\iota(A_k)=\iota(B_k)=k$ and every arch has minimal length $|\ar_{A_k}(i)|=|\ar_{B_k}(i)|=|\Sigma|$ for all $k\in\mathbb N$ and $i\in [1:k]$. Furthermore the $\sas$ of $B_k$ are exactly the monotonically decreasing sequences of length $k+1$, whereas no $\sas$ of $A_k$ has a strictly monotonically increasing or decreasing factor of length $3$. We will devote the rest of this chapter to an analysis of the sets $\sas(w)$ and $\mas(w)$ for $w$ being either $A_k$ or $B_k$. It turns out that $B_k$ has a small (that is polynomial in $\len{B_k}$) amount of $\sas$ but the amount of $\mas$ is exponential (in $\len{B_k}$)  whereas every $\mas$ of $A_k$ is an $\sas$ and $A_k$ has the maximal amount of $\sas$ among all words $w$ with universality index $\iota(w)=k$. We start with the following result.\looseness=-1

\begin{proposition}\label[proposition]{prop:decreasing}
The word $B_k$ has a polynomial number of $\sas$ and exponentially (in the length of the word) more $\mas$ than $\sas$.
\end{proposition}

\begin{proof}
Recall our assumption that $\sigma\geq 4$ is a large even constant, and let $k=(2\sigma+2)m$ where $m$ is a natural number.
Let $v$ be an $\sas$ of $B_k$.
From \cref{lem:sas},
it follows that $v[i]\geq v[i+1]$,
for all $i\in [1:|v|]$.
So, $v$ is a decreasing sequence of numbers from $\{1,\ldots,\sigma\}$, of length $k+1$.
Let us count these sequences.
We first count the number of sequences whose elements are contained in a fixed set $T\subset \{1,\ldots,\sigma\}$ of size $\len T=t$.
For each $T$ there are exactly $\binom{k}{t-1}$ such sequences
(as we essentially need to choose,
for each value except the smallest one, of the $t$ values,
the rightmost element of the sequence which takes that certain value and $k+1$ is necessarily reserved for the smallest value).
Moreover, we can choose the $t$ values taken by the numbers in the sequence in $\binom{\sigma}{t}$ ways.
So, the number of decreasing sequences of numbers from $\{1,\ldots,\sigma\}$, of length $k+1$,
is $S_{k+1}=\sum_{t=1}^\sigma \left( \binom{\sigma}{t} \cdot \binom{k}{t-1} \right).$
\noindent
Since $t-1 < \sigma\le \frac{k}{2}$ we have $\binom{k}{t-1}\le \binom{k}{t}$ and therefore $$S_{k+1}\leq \sum_{t=1}^\sigma \left( \binom{\sigma}{t} \cdot \binom{k}{t} \right) \leq \sum_{t=1}^\sigma \left ( \left (\frac{e \sigma }{t}\right)^t \cdot \left (\frac{ek}{t}\right)^t \right)=\sum_{t=1}^\sigma \frac{(e^2 \sigma k)^t}{t^{2t}} .$$
\noindent
Applying Chebyshev's sum inequality we obtain that the following holds: $$S_{k+1}\leq \frac{1}{\sigma}\cdot \left (\sum_{t=1}^\sigma (e^2 \sigma k)^t\right) \cdot \left (\sum_{t=1}^\sigma \frac{1}{t^{2t}}\right) .$$
\noindent
Thus, \vspace*{-1mm}
$$S_{k+1}\leq \frac{1}{\sigma} \cdot \frac{(e^2 \sigma k)^{\sigma+1}-1}{(e^2 \sigma k)-1} \cdot \frac{\pi^2}{6} \leq 4\cdot \frac{(e^2\sigma k)^\sigma}{\sigma} .$$
\noindent
As $\sigma$ is assumed to be constant, it follows that $S_{k+1}$ is bounded by a polynomial from above. We can also prove a lower bound for $S_{k+1}$. We have that:
\begin{align*}
S_{k+1}&\geq \sum_{t=1}^\sigma \left ( \left (\frac{ \sigma }{t}\right)^t \cdot \left (\frac{k}{t-1}\right)^{t-1} \right)\geq \sum_{t=1}^\sigma \left ( \left (\frac{ \sigma }{t}\right)^{t-1} \cdot \left (\frac{k}{t}\right)^{t-1}\right)\\ &=\sum_{t=1}^\sigma \left(\frac{(\sigma k)}{t^{2}}\right)^{t-1} \geq \sum_{t=1}^\sigma \left(\frac{(\sigma k)}{\sigma^{2}}\right)^{t-1}.
\end{align*}
Thus, \vspace*{-1mm}
$$S_{k+1}  \geq \frac{( k/\sigma)^{\sigma} -1}{(k/\sigma) -1}\geq (k/\sigma)^{\sigma-1}.$$
\noindent
A family of words from $\mas(B_k)$, which are not in $\sas(B_k)$, are the words $$v=\left( \Pi_{r\in [1:m]} (\sigma \cdot u_r \cdot 1\cdot \sigma)\right) \sigma,$$ where $\Pi_{r\in [1:m]}$ is used to denote concatenation of the terms and, for all $r$, $u_r$ is a decreasing sequence of length $2\sigma$ of numbers from $\{1,\ldots,\sigma\}$ not ending in $\sigma$, that is $u_i$ is not the constant sequence consisting in $2\sigma$ occurrences of the letter $\sigma$.

\medskip
We will show that $v$ is a $\mas$ of $B_k=v_1^k$ by showing that it satisfies \cref{thm:mas}:
since the $u_r$ are decreasing and $v_1$ is strictly increasing, the only factor of length $2$ of $v$ which is a subsequence of $v_1$ is $1\cdot\sigma$. Hence when searching for positions $i_1,i_2,\ldots i_{\len v -1}$ satisfying \cref{thm:mas:abs} positions $i_j$ and $i_{j+1}$ can be found inside the same arch of $B_k$ if and only if $v[j:j+1] = 1\cdot\sigma$. So, for $1\leq\ell\leq m$ the prefix $\Pi_{r\in[1:\ell]}(\sigma\cdot u_r\cdot 1\cdot\sigma)$ of $v$ is a subsequence of the prefix $v_1^{(\len{u_r}+2)\ell}$ of $B_k$ but absent from any shorter prefix of $B_k$. Since $B_k = v_1^k =v_1^{(2\sigma+2)m}= v_1^{(\len{u_r}+2)m}$, $v$ is absent from $B_k$ and its prefix of length $\len v-1$ is a subsequence of $B_k$. Furthermore, as every letter of $\Sigma$ occurs exactly once in $v_1$, there is only one possibility to choose indices $i_1,i_2,\ldots, i_{\len v-1}$ satisfying \cref{thm:mas:abs} of \cref{thm:mas}, which is, $i_1=\sigma$ and consecutive positions $i_j,i_{j+1}$ are uniquely chosen from consecutive arches of $B_k$ unless $v[j]=1$ and $v[j+1]=\sigma$ in which case $i_j$ and $i_{j+1}$ are uniquely chosen from the same arch. We will show that these indices satisfy \cref{thm:mas:v1,thm:mas:in,thm:mas:notin} of \cref{thm:mas} as well. Firstly $v[1]=\sigma$ does not occur in $B_k[1:i_1-1]=1\cdot 2\cdots\sigma-1$. Secondly recall that $v[j]>v[j-1]$ if and only if $v[j-1]=1$ and $v[j]=\sigma$, thus
$$v[j]\notin\al(B_k[i_{j-1}+1:i_j-1])=
\begin{cases}
\Sigma\setminus [v[j]:v[{j-1}]],\quad
&\text{if }v[j]\leq v[j-1]\\
\{2,3,\ldots, \sigma-1\},
&\text{if }v[j]> v[j-1].
\end{cases}$$
For the last item of \cref{thm:mas} note that if $v[j-1]\leq v[j-2]$ then either $v[j]\leq v[j-1]$ or $v[j]=\sigma >u_i[\len{u_i}]=v[j-2]$. Otherwise, if $v[j-1]> v[j-2]$, then $v[j]=\sigma$. Thus
$$v[j]\in\al(B_k[i_{j-2}+1:i_{j-1}])=
\begin{cases}
\Sigma\setminus [v[{j-1}]+1:v[{j-2}]],\quad
&\text{if }v[j-1]\leq v[j-2]\\
\{2,3,\ldots, \sigma\},
&\text{if }v[j-1]> v[j-2].
\end{cases}$$

Based on the results shown above about $\sas$ and the assumption that $\sigma \geq 4$, the size $S'$ of the family of words $v$ as defined above fulfils:
$$S'\geq (S_{2\sigma}-1)^m \geq \left(\left(\frac{2\sigma-1}{\sigma}\right)^{\sigma-1}-1\right)^m\geq 4^m.$$

Given that $m$ is not a constant, but a parameter which grows proportional to $|B_k|$, the statement of the theorem follows.
\end{proof}

Based on \cref{thm:mas,lem:sas}, it follows that the $\sas$ of $B_k$ correspond to decreasing sequences of length $k+1$ of numbers from $\{1,\ldots,\sigma\}$. A family of words included in $\mas(B_k)$, which is disjoint from $\sas(B_k)$, consists of the words $v=\left( \Pi_{t=1,m} (\sigma \cdot u_i \cdot 1\cdot \sigma)\right) \sigma$, where $u_i$ is a decreasing sequence of length $2\sigma$ of numbers from $\{1,\ldots,\sigma\}$, for all $i$. By counting arguments, we obtain the result stated above. However, there are words whose sets of $\mas$ and $\sas$ coincide.\looseness=-1
\begin{proposition}\label[proposition]{propAlternateone}
$\mas(A_k)=\sas(A_k)$.
\end{proposition}

\begin{proof} Assume that there exist an $\mas$ $v$ of $A_k$ which is not an $\sas$ (i.e., $|v|=m$ is strictly greater than $k+1$). We go left-to-right through $A_k$ and greedily choose positions $ i_1 < \ldots < i_m\le k\sigma $ such that $v[1]\cdots v[m] = A_k[i_1]\cdots A_k[i_m]$ and $i_j$ is the leftmost occurrence of $A_k[i_j]$ in $A_k[i_{j-1}+1:k\sigma]$. Let $a_j=A_k[i_j]$ for $j\in [1:m]$. As $m-1>k$, there exists an arch $\ar_{A_k}(\ell)$ of $A_k$ which contains at least two positions $i_{\ell}$ and $i_{\ell+1}$, and all archs  $\ar_{A_k}(\ell')$ with $\ell'<\ell$ contain exactly one such position. We can assume without loss of generality that $\ell$ is odd (the even case is identical). If $\ell>1$ then it is easy to see that $a_{\ell-1}\leq a_\ell$ (otherwise $a_{\ell}$ would be in the same arch as $a_{\ell-1}$) and $a_\ell<a_{\ell+1}$ (as $a_{\ell+1}$ is to the right of $a_{\ell}$ in an odd arch). So, $a_{\ell-1}<a_{\ell+1}$. It follows that when removing $v[\ell]$ from $v$, we obtain a word which is not a subsequence of $A_k$. If $\ell=1$ we get the same conclusion, trivially.
This is a contradiction to the fact that $v$ is an $\mas$. So our assumption was wrong, and $v$ must be an $\sas$.
\end{proof}


We can actually show a slightly stronger property: any word which has no proper $\sas$ is a permutation of $A_k$ for a suitable $k$. For the rest of this chapter $\pi$ always denotes a morphism of $\Sigma^\ast$ that acts as a permutation on the letters of $\Sigma$, that is, $\pi$ is a bijection of $\Sigma$ into itself and satisfies $\pi(uv)=\pi(u)\pi(v)$ for all $u,v\in\Sigma^\ast$. For simplicity we call $\pi$ a permutation (of $\Sigma$). In contrast we denote any permutation of integers, e.g. positions of a word, by $\theta$. For illustration let $\pi$ and $\theta$ both denote the permutation that maps $1$ to $2$ and vice versa, then $\pi(1322) = 2311$ and $\theta(1322) = 3122$.

\vspace{1.8mm}
The first observation follows directly.

\begin{lemma}\label[lemma]{perm}
If $\pi$ is a permutation of $\Sigma$ and $w\in\Sigma^*$ then $\pi(v)$ is an $\sas$ of $\pi(w)$ if and only if $v$ is an $\sas$ of $w$.
Hence $|\sas(w)|$ = $|\sas(\pi(w))|$.
\end{lemma}

\begin{proposition}
Let $w\in\Sigma^*$, then $|\sas(w)| = |\mas(w)|$ if and only if there is a permutation $\pi$ such that $\pi(w) = A_{\iota(w)}$.
\end{proposition}

\begin{proof}
The if-part is immediate: $|\sas(w)|=|\sas(\pi(w))|=|\mas(\pi(w))|=|\mas(w)|$ by \cref{propAlternateone}.

\medskip
For the converse let $w\in\Sigma^*$ satisfying $|\sas(w)| = |\mas(w)|$. Firstly assume there is $a\in\Sigma$ such that $|w|_a > \iota(w)$. Then $a^{|w|_a+1}$ is an $\mas$ of $w$ but not an $\sas$ of $w$, a contradiction. So every letter occurs exactly $\iota(w)$ times in $w$ and therefore exactly once in each arch $\ar_w(\ell)$.
Hence~there is a permutation $\pi$ such that $\pi(\ar_w(1)) = v_1$.
Now assume there is an integer $\ell\leq\iota(w)$ such that $\pi(\ar_w(1)\ar_w(2)\cdots\ar_w(\ell)) \neq A_{\ell}$. Let $k$ be minimal with this property.
By definition of $\pi$ we have $k>1$ and w.l.o.g we can assume $k$ to be even.
Then $\pi(\ar_w(k))\neq v_2 = \sigma\cdots 21$, that is we can find two positions $i<j\in\ar_w(k)$ such that $\pi(w[i])<\pi(w[j])$.
We claim that $u =w[i]^kw[j]^{\iota(w)-k+2}$ is an $\mas$ but not an $\sas$ of $w$.
We start by checking if $u_1=w[i]^{k-1}w[j]^{\iota(w)-k+2}$ and $u_2=w[i]^kw[j]^{\iota(w)-k+1}$ are subsequences of $w$.
We have $\pi(\ar_w(k-1))=v_1=12\cdots\sigma$ and $\pi(w[i])<\pi(w[j])$ hence $w[i]w[j]$ occurs in $\ar_w(k-1)$.
Then $w[i]^{k-1}w[j]$ is a subsequence of $\ar_w(1)\cdots\ar_w(k-1)$ and $w[j]^{\iota(w)-k+1}$ is a subsequence of $\ar_w(k)\cdots\ar_w(\iota(w))$.
Similarly $w[i]^kw[j]$ is a subsequence of $\ar_w(1)\cdots\ar_w(k)$ and $w[j]^{\iota(w)-k}$ is a subsequence of $\ar_w(k+1)\cdots\ar_w(\iota(w))$.
So $u_1,u_2$ are subsequences of $w$ and $u$ is an $\mas$ of $w$.
We have $\len u = \iota(w)+2$ and so $u$ is not an $\sas$ of $w$, a contradiction.
Hence $\pi(\ar_w(1)\!\cdots\ar_w(\ell))\!=\!A_\ell$ holds for all $\ell\!\leq\iota(w)$ and therefore $\pi(w)\!=\!A_{\iota(w)}$.
\end{proof}

In particular, one can show that the number of $\sas$ in the words $A_k$ is exponential in the length of the word. The main idea is to observe that an $\sas$ in $A_k$ is a sequence $i_1,\ldots,i_{k+1}$ of numbers from $\{1,\ldots,\sigma\}$, such that $i_{2\ell-1}\geq i_{2\ell}$ for all $\ell\in [\, 1:\lceil k/2 \rceil \, ]$ and $i_{2\ell} \leq i_{2\ell+1}$ for all $\ell\in [\, 1:\lfloor k/2 \rfloor \, ]$. We then can estimate the number of such sequences of numbers and obtain the following result.
\begin{proposition}\label[proposition]{prop:alternate2}
The word $A_k$ has an exponential (in $|A_k|$) number of $\sas$.
\end{proposition}

\begin{proof}
Again, recall our assumption that $\sigma$ is a large even constant. An $\sas$ in $A_k$ is a sequence $i_1,\ldots,i_{k+1}$ of numbers from $\{1,\ldots,\sigma\}$, such that $i_{2\ell-1}\geq i_{2\ell}$ for all $\ell\in [\, 1:\lceil k/2 \rceil \, ]$ and $i_{2\ell} \leq i_{2\ell+1}$ for all $\ell\in [\, 1:\lfloor k/2 \rfloor \, ]$. Let $S_{t,i}$ be the number of strings of length $t\leq k$ ending with letter $i\in \Sigma$, which can be a prefix of such a sequence (i.e., of an $\sas$ of $A_k$).
\eject

The first observations are the following:
\begin{itemize}
\itemsep=0.9pt
\item For even $t\geq 2$ we have that $S_{t,i}=\sum_{j\in [i:\sigma]} S_{t-1,j}$, and, therefore, $S_{t,1}=\sum_{i\in [1:\sigma]} S_{t-1,i}$. The sequence $S_{t,\sigma}, S_{t,\sigma-1} , S_{t,\sigma-2},\ldots, S_{t,1}$ is increasing.
\item For odd $t\geq 3$ we have that $S_{t,i}=\sum_{j\in [1:i]} S_{t-1,j}$, so $S_{t,\sigma}=\sum_{i\in [1:\sigma]} S_{t-1,i}$. The sequence $S_{t,1}, S_{t,2} , \ldots, S_{t,\sigma-1}, S_{t,\sigma}$ is increasing.
\end{itemize}

Clearly, $S_{1,i}=1$ for all $i\in \Sigma$.

We claim that for all $t$ we have that $\sum_{i\in [1:\sigma]} S_{t,i} \geq \left ( \frac{\sigma}{2}\right )^t.$ We will prove this by induction.

This clearly holds for $t=1$. Assume that it holds for $t\leq  q$ and we show that it holds for $t=q+1$.

First, assume that $q+1$ is odd. By the induction hypothesis, we have that  $S_{q+1,\sigma}=\sum_{i\in [1:\sigma]} S_{q,i}\geq \left ( \frac{\sigma}{2} \right )^{q}.$ Moreover, $S_{q+1,1}=\sum_{i\in [1:1]} S_{q,i}=S_{q,1}= \sum_{i\in [1:\sigma]} S_{q-1,i}\geq \left ( \frac{\sigma}{2} \right )^{q-1}.$

\medskip
Now, we have that:
$$\sum_{i\in [1:\sigma]} S_{q+1,i} = \sum_{j\in[1:\sigma/2]} \left(S_{q+1,j}+S_{q+1,\sigma+1-j}\right) = \sum_{j\in [1:\sigma/2]} \left ( \sum_{g\in [1:j]} S_{q,g} + \sum_{g\in [1:\sigma-j+1]} S_{q,g} \right ), $$
and, therefore,
$$\sum_{i\in [1:\sigma]} S_{q+1,i} = \sum_{j\in [1:\sigma/2]} \left ( S_{q,1} + \sum_{g\in [2:j]} S_{q,g} + \sum_{g\in [1:\sigma-j+1]} S_{q,g} \right ). $$

Now, as $S_{q,\sigma}, S_{q,\sigma-1} , S_{q,\sigma-2},\ldots, S_{q,1}$ is increasing, we get that:
$$\sum_{i\in [1:\sigma]} S_{q+1,i} \geq \sum_{j\in [1:\sigma/2]} \left ( S_{q,1} + \sum_{g\in [\sigma-j+2:\sigma]} S_{q,g} + \sum_{g\in [1:\sigma-j+1]} S_{q,g} \right ). $$
Therefore:
$$\sum_{i\in [1:\sigma]} S_{q+1,i} \geq \sum_{j\in [1:\sigma/2]} \left ( S_{q,1} + \sum_{g\in [1:\sigma]} S_{q,g} \right )=\frac{\sigma}{2} \left(S_{q,1} + \sum_{g\in [1:\sigma]} S_{q,g} \right ). $$

We can now apply the induction hypothesis, and obtain that
$$\sum_{i\in [1:\sigma]} S_{q+1,i} \geq \frac{\sigma}{2} \left(\left(\frac{\sigma}{2}\right )^{q-1} + \left(\frac{\sigma}{2}\right )^{q} \right )> \left(\frac{\sigma}{2}\right )^{q+1}.$$

An analogous argument works for the case when $q+1$ is even.

\noindent This concludes our induction proof.

Therefore, the number of strings of length $k$ which can be a prefix of a $\sas$ of $A_k$ is at least  $\left(\frac{\sigma}{2}\right )^{k}$. We note that this is not a tight bound, and a more careful analysis should definitely improve it.

So $|\sas(A_k)|\geq \left(\frac{\sigma}{2}\right )^{k}$, which, given that $\sigma$ is a constant, proves our statement: $|\sas(A_k)|\geq \left(\frac{\sigma}{2}\right )^{|A_k|/\sigma}.$
\end{proof}

The following few results formalise an additional insightful observation about the word $A_k$ and its set of $\sas$.
We will need the following initial remarks and lemmas.
\eject

\begin{lemma}\label[lemma]{lem:cat}
Let $w,w_1,w_2\in\Sigma^*$ such that $w=w_1w_2$, $\iota(w_2)\ge 1$, and $r(w_1)=\varepsilon$.

Then $v[1]\cdots v[\iota(w)+1]$ is an $\sas$ of $w$ if and only if $v[1]\cdots v[\iota(w_1)+1]$ is an $\sas$ of $w_1$
and $v[\iota(w_1)+1]\cdots v[\iota(w)+1]$ is an $\sas$ of $w_2$.
\end{lemma}
\begin{proof}
%
Note first that the arches of $w$ coincide with the arches of $w_1$ and $w_2$.
If $v\in\sas(w)$ and we greedily choose positions for $v[1]\cdots v[\iota(w_1)]$ in $w_1$, we find the same positions as if we greedily choose them in $w$.
As $r(w_1)=\varepsilon$, $v[1]\cdots v[\iota(w_1)+1]$ is an $\sas$ of $w_1$. Furthermore, let $i$ be the position of $v[\iota(w_1)+1]$ in $\ar_{w_2}(1)$. Now $v[\iota(w_1)+2]$ occurs in $\ar_{w_2}(1)[1:i]$ because it occurs in $\ar_{w_2}(1)$ and if it occurs to the right of $i$ then, by \cref{lem:sas}, $v$ is not an $\sas$ of $w$. Hence we can find positions for $v[\iota(w_1)+1]$ and $v[\iota(w_1)+2]$ in $w_2$ satisfying \cref{thm:mas}. For the remaining letters of $v$ we can find positions in $w_2$ satisfying \cref{thm:mas} because $w_2$ is a suffix of $w$ and we can find these positions in $w$. Hence $v[\iota(w_1)+1:\iota(w)+1]$ is an $\sas$ of $w_2$.

\medskip
The other implication is trivial.
\end{proof}

Generally, we apply \cref{lem:cat} for $w_1=\ar_w(1)\cdots\ar_w(\ell)$, $w_2=\ar_w(\ell+1)\cdots\ar_w(\iota(w))r(w)$, for some $\ell\le\iota(w)$.

We continue with recalling the so-called rearrangement inequality.
\begin{lemma}\label[lemma]{riem}
If $a_1,\ldots,a_n$ and $b_1,\ldots,b_n$ are both increasing (respectively, decreasing) sequences of natural numbers and $\theta$ is a permutation of $\{1,2,\ldots n\}$,
then $\sum_{i=1}^na_ib_{n+1-i}\le\sum_{i=1}^na_ib_{\theta(i)}\le \sum_{i=1}^na_ib_i$.
\end{lemma}

The following lemma is immediate.
\begin{lemma}\label[lemma]{sub}
If $w'$ is a subsequence of $w$ and $\iota(w)=\iota(w')$, then $\sas(w)\subset\sas(w')$.
\end{lemma}

Before we begin to show the main observation on the word $A_k$ please recall $v_1=12\cdots\sigma$, $v_2=\sigma(\sigma-1)\cdots 1$ and
$A_k=w_1w_2\cdots w_{k-2} w_{k-1}w_k$, where $w_\ell=v_1$ if $\ell$ is odd and $w_\ell=v_2$ otherwise.
We note $\sas(v_1)=\{ij\in\Sigma^2\mid i\ge j\}$ and $\sas(v_2)=\{ij\in\Sigma^2\mid i\le j\}$.
That is we have $i~\sas$ of $v_1$ starting with $i$ and $\sigma+1-i~\sas$ of $v_1$ ending on $i$.
For $v_2$ these values are switched: we have $\sigma+1-i~\sas$ starting with $i$ and $i~\sas$ of $v_2$ ending on $i$.

\begin{proposition}\label[proposition]{maxsas}
$|\sas(A_k)|\geq |\sas(w)|$ holds for all  $w\in \Sigma^\ast$ with $\iota(w)=k$.
\end{proposition}

\begin{proof}
By \cref{sub} it suffices to show that \cref{maxsas} holds for words $w$,
such that every arch of $w$ contains every letter exactly once and $r(w)=\varepsilon$.
So, in this proof, all words we use have each letter of the alphabet exactly once in each arch. This means for any $1$-universal word $v$ there is a permutation $\pi$ such that $\pi(v)=v_2$.

We will show \cref{maxsas} in the following way:
if $\iota(w)=1$ then there is a permutation $\pi$ of $\Sigma$ such that $\pi(w)=v_1=A_1$, hence $|\sas(w)|=|\sas(A_1)|$. For general $\iota(w)\ge 2$ we let
 $w_\ell=\ar_w(1)\ar_w(2)\cdots \ar_w(\iota(w) - \ell)$ (with $\iota(w_\ell)=\iota(w)-\ell$) and show that there exists a permutation $\pi$ of $\Sigma$ such that $|\sas(w)|\le |\sas(\pi(w_\ell)A_\ell)|$ for every word $w$ and integer $2\leq\ell\leq \iota(w)$.
Inductively this shows \cref{maxsas}.

\eject
Firstly let $f_i(w)=|\Sigma^*i\cap\sas(w)|,g_i(w)=|i\Sigma^*\cap\sas(w)|$ and $h_{i,j}(w)=|i\Sigma^*j\cap\sas(w)|$. That is, $f_i,g_i,h_{i,j}$ denote the respective number of $\sas$ of $w$ which end on $i$,
start with $i$, and start with $i$ and end with $j$ respectively.
Then, for any factorisation $w=u_1u_2$ satisfying \cref{lem:cat} we have $|\sas(w)| = |\sas(u_1u_2)| = \sum_{i\in\Sigma}f_i(u_1)g_i(u_2)$.
For $w\in\Sigma^*$ and $v_1,v_2,A_k$ as above, we have $f_i(wv_1) = \sum_{j\ge i} f_j(w)$
and $g_i(v_2A_k)=\sum_{j\ge i} g_j(A_k)$ because $h_{j,i}(v_1) = h_{i,j}(v_2) = 1$ if $j\ge i$ and $0$ otherwise.
Furthermore $f_i(w) = f_{\pi(i)}(\pi(w))$, $g_i(w) = g_{\pi(i)}(\pi(w))$ and $h_{i,j}(w) = h_{\pi(i),\pi(j)}(\pi(w))$ by \cref{perm}; when using this observation, we will refer to it using $\mbox{(*)}$.

\medskip
We start by showing that the proposition holds for $\ell=2$. Let $\iota(w)=k\ge 2$, $\pi$ be a permutation of $\Sigma$ such that $\pi(\ar_w(k-1))=v_1$ and $v'=\pi(\ar_w(k))$ then
{\small{
\begin{align*}
|\sas(w)| &= |\sas(\pi(w))| = |\sas(\pi(w_2)v_1v')|&\\
&= \sum_{i\in\Sigma} f_i(\pi(w_2)v_1)g_i(v')=\sum_{i\in\Sigma} f_i(\pi(w_2)v_1) g_{\tau(i)}(v_2)&\\
	&\le \sum_{i\in\Sigma} f_i(\pi(w_2)v_1) g_i(v_2) = |\sas(\pi(w_2)v_1v_2)|& \mbox{(by \cref{riem})}\\
		&=|\sas(\pi(w_2)A_2)|,&
\end{align*}   } }
where $\tau$ denotes the permutation which maps $v'$ to $v_2$. Please note that $g_i(v_2)=\len{\{ij\mid i,j\in\Sigma,j\ge i\}}$ and $f_i(\pi(w_2)v_1)=\sum_{j\ge i} f_j(\pi(w_2))$ are decreasing (as sequences w.r.t. $i$).

\medskip
Now, we assume we have that for every $w\in\Sigma^*$ and $\ell<\iota(w)$ we can find a permutation $\pi$ of $\Sigma$ such that $|\sas(w)|\le |\sas(\pi(w_\ell)A_\ell)|$. We show that the result holds for $\ell+1$ as well; that is, we show that  $|\sas(w)|\le |\sas(\pi''(w_{\ell+1})A_{\ell+1})|$ for some permutation $\pi''$. Indeed, we have:
{\small{
\begin{align*}
|\sas(w)| &\leq |\sas(\pi(w_\ell)A_\ell)| = |\sas(\pi(w_{\ell+1})\pi(\ar_w(\iota(w) - \ell))A_\ell)|& \\[-1pt]
     &= \sum_{i\in\Sigma} f_i(\pi(w_{\ell+1}))\left(\sum_{j\in\Sigma}h_{i,j}(\pi(\ar_w(\iota(w) - \ell))) g_{j}(A_\ell)\right)& \\[-1pt]
	 &= \sum_{i\in\Sigma} f_i(\pi(w_{\ell+1}))\left(\sum_{j\in\Sigma}h_{\tau(i),\tau(j)}(v_2) g_{j}(A_\ell)\right)&\mbox{(by (*))}\\[-1pt]
	 &\leq \sum_{i\in\Sigma} f_i(\pi(w_{\ell+1}))\left(\sum_{j\in\Sigma}h_{\tau(i),j}(v_2) g_{j}(A_\ell)\right)&\mbox{(by \cref{riem})} \\[-1pt] 
	 &= \sum_{i\in\Sigma} f_{\tau(i)}(\tau(\pi(w_{\ell+1})))\left(\sum_{j\in\Sigma}h_{\tau(i),j}(v_2) g_{j}(A_\ell)\right)&\mbox{(by (*))}\\[-1pt]
	 &= \sum_{\tau(i)\in\Sigma} f_i(\tau(\pi(w_{\ell+1})))\left(\sum_{j\in\Sigma}h_{i,j}(v_2) g_{j}(A_\ell)\right)& \\[-1pt]
	 &= \sum_{i\in\Sigma} f_i(\tau(\pi(w_{\ell+1})))\left(\sum_{j\in\Sigma}h_{i,j}(v_2) g_{j}(A_\ell)\right)& \\[-1pt]
	 &= |\sas(\tau(\pi(w_{\ell+1}))\pi'(A_{\ell+1}))| = |\sas(\pi''(w_{\ell+1})A_{\ell+1})|,& \\[-1pt]
\end{align*}  } }

\noindent where $\tau$ is a permutation mapping $\ar_w(\iota(w) - \ell)$ to $v_2$, $\pi'$ is the permutation which maps $i$ to $\sigma+1-i$ (and therefore also maps $v_1$ to $v_2$) and $\pi''=\left(\pi'\right)^{-1}\circ\tau\circ\pi$.
The second inequality in the reasoning above follows from \cref{riem} because $g_{j}(A_k)$ and $h_{\tau(i),j}(v_2)$ both are monotonically increasing as sequences in $j$.
This concludes our proof.
\end{proof}

In the conclusion of this section, we make the following remarks. \Cref{prop:decreasing,prop:alternate2} motivate our investigation for compact representations of the sets of $\sas$ and $\mas$ of words. These sets can be exponentially large, and we would still like to have efficient (i.e., polynomial) ways of representing them, allowing us to explore and efficiently search these sets.

\section{Algorithms} \label{sec:algorithms}

The results we present from now on are of algorithmic nature. The computational model we use to describe our results is the standard unit-cost RAM with logarithmic word size: for an input of size $n$, each memory word can hold $\log n$ bits. Arithmetic and bitwise operations with numbers in $[1:n]$ are, thus, assumed to take $O(1)$ time. Numbers larger than $n$, with $\ell$ bits, are represented in $O(\ell/\log n)$ memory words, and working with them takes time proportional to the number of memory words on which they are represented. In all the problems, we assume that we are given a word $w$, with $|w|=n$, over an alphabet $\Sigma=\{1,2,\ldots,\sigma\}$, with $|\Sigma|=\sigma\leq n$. That is, we assume that the processed words are sequences of integers (called letters or symbols, each fitting in $O(1)$ memory words). On the one hand, note that we no longer assume that $\sigma$ is a constant, as in the combinatorial results of the previous section. On the other hand, assuming that $\sigma\in O(n)$ is a common assumption in string algorithms: the input alphabet is said to be {\em an integer alphabet}, see also the discussion in, e.g.,~\cite{crochemore}. For its use in the particular case of algorithms related to subsequences see, for instance, \cite{mfcs2020,DayFKKMS21}.

In all the problems, we assume that we are given a word $w$ of length $n$ over an {\em integer alphabet} $\Sigma = \{1, 2, \ldots, \sigma\}$, with $\len \Sigma = \sigma \leq n$.
As the problems considered here are trivial for unary alphabets, we also assume $\sigma\geq 2$.

We start with some preliminaries and simple initial results.
The decomposition of word $w$ into its arches can be done with a greedy approach.
The following theorem is well known and a proof can be seen for example in \cite{Barker2020}.
It shows that the universality index and the decomposition into arches can be obtained in linear time. \looseness=-1
\begin{theorem}\label[theorem]{thm:archfact-algo}
	Given a word $w$, of length $n$, we can compute the universality index $\iota(w)$, the arch factorisation $\ar_w(1)\cdots\ar_w(\iota(w))\rest(w)$ of $w$, as well as the set $\Sigma\setminus \al(r(w))$ of letters which do not occur in $\rest(w)$ in linear time $O(n)$.
\end{theorem}

\begin{proof}
One can compute greedily
the decomposition of $w$ into arches $w = \ar_w(1) \cdots \ar_w(k) \rest(w)$ in linear time
as follows.
\begin{itemize}
	\item $\ar_w(1)$ is the shortest prefix of $w$ with $\al(\ar_w(1)) = \Sigma$,
		or $\ar_w(1) = w$ if there is no such prefix;
	\item if $\ar_w(1) \cdots \ar_w(i) = w[1:t]$, for some $i \in [1:k]$ and $t \in [1:n]$,
	we compute $\ar_w(i+1)$ as the shortest prefix of $w[t+1:n]$ with $\al(\ar_w(i+1)) = \Sigma$,
	or $\ar_w(i+1) = w[t+1:n]$ if there is no such prefix.
\end{itemize}

The concrete computation can be also seen in \cref{alg:archfact}. Note that the complexity is linear because the re-initialization of $C$ is amortized by the steps in which the elements of $C$ were incre-\linebreak mented.

\begin{algorithm}[!htb]
	\SetAlgoLined
	\KwIn{Word $w$ with $\len w = n$, alphabet $\Sigma = \al(w)$ with $\len \Sigma = \sigma$}
	\KwOut{Variables $n_1, \ldots, n_k$ containing the last position of the respective arch}
	
	define array $C$ of length $\sigma$ with elements initially set to $0$\;
	\tcc{Array C is used for marking letters as already encountered}
	define counter $h \gets \sigma$\;
	\tcc{Counter h is used to check if all letters of the alphabet have been encountered}
	define arch counter $k \gets 1$\;
	
	\For{$i = 1$ \KwTo $n$}{
		\If{$C[w[i]] == 0$}{
			$h \gets h - 1$\;
			$C[w[i]] \gets 1$\;
		}
		\If{$h == 0$}{
			$n_k \gets i$\;
			$k \gets k + 1$\;
			$h \gets \sigma$\;
			reset array $C$ by initialising all elements with $0$ again\;
		}
	}

	\Return $n_1$ \KwTo $n_k$\;

	\caption{computeArchFactorisation($w$, $\Sigma$)}
	\label{alg:archfact}
\end{algorithm}

Further, we can simply define a binary array $f$ indexed by $\Sigma$, whose components are initially set to $0$. We then go through $\rest(w)$, and set $f[a]=1$ if and only if $a$ appears in $\rest(w)$. Clearly, the set of letters which does not occur in $\rest(w)$ (that is, $\Sigma\setminus \al(r(w))$) is the set of letters $a$ with $f[a]=0$. A list with these letters can be computed in $O(\sigma)$ time. We now only have to note that this processing takes at most $O(n)$ time. This concludes our proof.
\end{proof}

\begin{remark}
For a word $w$ with length $n$,
further helpful notations are $\nextpos_w(a,i)$, which denotes the next occurrence of letter $a$ in the word $w[i:n]$,
and $\lastpos_w(a,i)$, which denotes the last occurrence of letter $a$ in the word $w[1:i]$.
Both these values can be computed by simply traversing $w$ from position $i$ to the right or left, respectively.
The runtime of computing $\nextpos_w(a, i)$ is proportional to the length of the shortest factor $w[i:j]$ of $w$ such that $w[j]=a$
as we only traverse the positions from word $w$ once from $i$ towards the right until we meet the first letter $a$.
The runtime of $\lastpos_w(a,i)$ is proportional to the length of the shortest factor $w[j:i]$ of $w$ such that $w[j]=a$
as we only traverse the positions from word $w$ once from $i$ towards the left until we meet the first letter $a$.
\end{remark}

\begin{proof}
See \cref{alg:next,alg:last}.
\SetKw{KwOr}{or}\medskip

\begin{algorithm}[!htb]\label{alg:next}
	\SetAlgoLined
	\KwIn{Word $w$ with $|w| = n$, letter $a$, position $i$}
	\KwResult{Position of next occurence from letter $a$ starting from position $i$ in word $w$}
	
	\If{$i < 1$ \KwOr $i > n$}{
		\Return $\infty$\;
	}
	
	\For{$j=i$ \KwTo $n$}{
		\If{$w[j] == a$}{
			\Return $j$\;
		}
	}
	
	\Return $\infty$\;
	
	\caption{$\nextpos_w(a, i)$}
\end{algorithm}

\SetKw{KwOr}{or}
\begin{algorithm}[!htb]
	\SetAlgoLined
	\KwIn{Word $w$ with $|w| = n$, letter $a$, position $i$}
	\KwResult{Position of last occurence from letter $a$ starting from position $i$ in word $w$}
	
	\If{$i < 1$ \KwOr $i > n$}{
		\Return $\infty$\;
	}
	
	\For{$j=i$ \KwTo $1$}{
		\If{$w[j] == a$}{
			\Return $j$\;
		}
	}
	
	\Return $\infty$\;
	
	\caption{$\lastpos_w(a,i)$}
	\label{alg:last}
\end{algorithm}

\vspace*{-3mm}
\end{proof}

\begin{algorithm}[!b]
	\SetAlgoLined
	\KwIn{Word $w$ with $|w| = n$, word $u$ with $|u|=m$ to be tested}
	
	pos $\gets 0$\;
	\For{$i=1$ \KwTo $m$}{
		pos $\gets \nextpos_w(u[i], pos+1)$\;
	}
	
	\Return pos == $\infty$ ? false : true\;
	
	\caption{isSubseq($w$,$u$)}
	\label{alg:isSubseq}
\end{algorithm}

Based on these notations, we can define the function $\isSubseq(w,u)$ which checks if a word $u$ is a subsequence of a word $w$.
This can be done by utilizing a greedy approach as discussed in \cref{combinatorics} and depicted in \cref{alg:isSubseq}.
While this approach is standard and relatively straightforward,
it is important to note it before proceeding with our algorithms.
The idea is the following.
We consider the letters of $u$ one by one,
and try to identify the shortest prefix $w[1:i_j]$ of $w$
which contains $u[1:j]$ as a subsequence.
To compute the shortest prefix $w[1:i_{j+1}]$ of $w$
which contains $u[1:j+1]$,
we simply search for the first occurrence of $u[j+1]$ in $w$ after position $i_j$.
The runtime of the algorithm $\isSubseq$ is clearly linear in the worst case.
When $u$ is a subsequence of $w$,
then $w[1:pos]$ is the shortest prefix of $w$ which contains $u$,
and the runtime of the algorithm is $O(pos)$.

\medskip
A further helpful notation is $\llo(w)=\min\{\lastpos_w(a,n)\mid a\in \Sigma\}$, the position of the leftmost of the last occurrences of the letters of $\Sigma$ in $w$. The following lemma is not hard to show.

\begin{lemma}\label[lemma]{lem:llo}
Given a word $w$ of length $n$, we can compute $\llo(w)$ in $O(n)$ time.
\end{lemma}

\begin{proof}
	See \cref{alg:llo}. This is clearly correct and runs in linear time.\medskip
	
	\begin{algorithm}[!htb]
		\SetAlgoLined
		\KwIn{Word $w$ with $\len w = n$, alphabet $\Sigma = \al(w)$ with $\len \Sigma = \sigma$}
		\KwOut{Leftmost position in word $w$ of all last occurences of the letters in $\Sigma$}
		
		define array $C$ of length $\sigma$ with elements initially set to $0$\;
		\tcc{Array C is used for marking letters as already encountered}
		define counter $h \gets \sigma$\;
		
		\For{$i=n$ \KwTo $1$}{
			\If{$C[w[i]] == 0$}{
				$C[w[i]] \gets 1$\;
				$h \gets h - 1$\;
				\If{$h == 0$}{
					\Return $i$\;
				}
			}
		}
		
		\Return $\infty$\;
		
		\caption{llo($w$, $\Sigma$)}
		\label{alg:llo}
	\end{algorithm}

\vspace*{-3mm}
\end{proof}

Based on these notions, we can already present our first results, regarding the basic algorithmic properties of $\sas$ and $\mas$. Firstly, we can easily compute an $\sas$ (and, therefore, an $\mas$) in a given word.

\begin{lemma}\label[lemma]{lem:oneSAS}
Given a word $w$ of length $n$ with $\iota(w)=k$, its arch decomposition (i.e., the last position of each arch),  and the set $\Sigma\setminus \al(r(w))$ (i.e., as a list of letters),
we can retrieve in $O(k)$ time an $\sas$ (and, therefore, an $\mas$) of $w$.
\end{lemma}

\begin{proof}
By \cref{archfact},
the word $m(w) =\ar_w(1)[|\ar_w(1)|] \cdots \ar_w(k)[|\ar_w(k)|]$ contains the unique last letters of each arch.
By appending a letter not contained in $\rest(w)$ to $m(w)$,
we achieve a word that is an $\sas$ satisfying the properties from \cref{lem:sas}.
\Cref{alg:getSAS} shows how to construct $m(w)$ in $O(k)$ time.
To extend $m(w)$ to an $\sas$, we have to append to the existing word a letter that is not contained in $\rest(w)$, which can be done $O(1)$ if that set is given as a list.

\begin{algorithm}[!htb]
	\SetAlgoLined
	\KwIn{Word with its arch factorization $w = \ar_w(1) \cdots \ar_w(k) \rest(w)$ with $|w| = n$}
	
	define empty word sas\;
	\For{$i=1$ \KwTo $k$}{
		add the last letter of arch $\ar_w(i)$ to the end of sas:
		
		sas = sas $\cdot \ar_w(i)[|\ar_w(i)|]$\;
	}
	
	add a letter of the alphabet that does not occur in $\rest(w)$ to the end of sas
	
	\Return sas\;
	\caption{getSAS($w$)}
	\label{alg:getSAS}
\end{algorithm}

\vspace*{-6mm}
\end{proof}

The following theorem shows that we can efficiently test if a given word $u$ is an $\sas$ or $\mas$ of a given word $w$.
\begin{theorem}\label[theorem]{lem:testSAS}
Given a word $w$ of length $n$
and a word $u$ of length $m$,
we can test in $O(n)$ time whether $u$ is an $\sas$ or $\mas$ of $w$.
\end{theorem}

\begin{proof}
	We recall that by definition a word $u$ is an $\sas$ of a word $w$ if and only if $|v| = \iota(w)+1$ and $v$ is absent.
	It is therefore sufficient to check the length of an $\sas$ candidate and whether it is a subsequence of $w$ or not.
	In \cref{alg:isSAS} we can now see this implemented with a runtime in $O(n)$.

\SetKw{KwAnd}{and}
\SetKwFunction{KwIsSubseq}{isSubseq}
\begin{algorithm}[!htb]
	\SetAlgoLined
	\KwIn{Word with its arch factorization $w = \ar_w(1) \cdots \ar_w(k) \rest(w)$,
		word $u$ (over the same alphabet as $w$) to be tested}
	
	\If{$\len u$ = $k+1$ \KwAnd $!$\KwIsSubseq{$w$, $u$}}{
		\Return true\;
	}\Else{
		\Return false\;
	}
	
	\caption{isSAS($w$,$u$)}
	\label{alg:isSAS}
\end{algorithm}

To decide if $u$ is an $\mas$,
we can check whether $u$ is indeed an absent subsequence
and every subsequence of $u$ of length $m-1$ is also a subsequence of $w$.
In other words, $u$ is turned into a subsequence of $w$ by removing any of its letters.
It is important to note at this point
that the subsequences of length $m-1$ of $u$ have the form $u[1:i]u[i+2:m]$ for some $i\leq m-1$.

Based on this observation,
we first compute using \cref{alg:isMAS} the shortest prefixes of $w$ containing $u[1:1], u[1:2], \ldots, u[1:m]$
and the shortest suffixes of $w$ containing $u[m:m], u[m-1:m], \ldots, u[1:m]$.
With the knowledge over these prefixes and suffixes,
we can quickly decide for a position $i$ of $u$,
if $u[1:i-1] u[i+1:m]$ is a subsequence or an absent subsequence of $w$.
For this, we check if the shortest prefix of $w$ containing $u[1:i-1]$ is overlapping with the shortest suffix containing $u[i+1:m]$.
If they are overlapping,
we obtained from $u$ an absent subsequence $u[1:i-1] u[i+1:m]$ by removing the letter $u[i]$,
meaning that $u$ is not an $\mas$.

\medskip
If, for all positions $i \in [1:m]$,
the shortest prefix containing $u[1:i-1]$
and the shortest suffix containing $u[i+1:m]$ are not overlapping,
then $u$ is clearly an $\mas$: all its subsequences of length $m-1$ are also subsequences of $w$,
while $u$ is absent.
Checking whether this property holds takes linear time.
\Cref{alg:isMAS} is an implementation in pseudocode of the strategy presented in the proof.
\SetKw{KwOr}{or}
\SetKwFunction{KwIsSubseq}{isSubseq}
\begin{algorithm}[!htb]
	\SetAlgoLined
	\KwIn{Word with its arch factorization $w = \ar_w(1) \cdots \ar_w(k) \rest(w)$ with $|w| = n$,
		word $u$ (over the same alphabet as $w$) to be tested with $|u| = m$}
	\BlankLine
	
	\lIf{$\len u < k+1$}{\Return false}
	\lIf{\KwIsSubseq{$w$,$u$}}{\Return false}
	\BlankLine
	
	isMAS $\gets true$\;
	define array L of size $m$\;
	define array R of size $m$\;
	\BlankLine
	
	\tcc{we want to fill the arrays L and R
		so that L[$i$] holds the end position of the shortest prefix of $w$ containing the subsequence $u[1:i]$,
		and R[$i$] holds the start position of the shortest suffix of $w$ containing the subsequence $u[i:m]$}
	L[$1$] $\gets \nextpos_w(u[1], 1)$\;
	\For{$i = 2$ \KwTo $m$}{
		L[$i$] $\gets \nextpos_w(u[i],$ L$[i-1]+1)$\;
	}
	\BlankLine
	
	R[$m$] $\gets \lastpos_w(u[m], n)$\;
	\For{$i = m-1$ \KwTo $1$}{
		R[$i$] $\gets \lastpos_w(u[i],$ R$[i+1]-1)$\;
	}
	\BlankLine
	
	\If{$!$\KwIsSubseq{$u[1:m-1]$} \KwOr $!$\KwIsSubseq{$u[2:m]$}}{
		\Return false\;
	}
	
	\For{$i = 1$ \KwTo $m-2$}{
		\If{L$[i] >= $R$[i+2]$}{
			\Return false\;
		}
	}
	\BlankLine
	
	\Return isMAS\;
	\caption{isMAS($w$,$u$)}
	\label{alg:isMAS}
\end{algorithm}
\end{proof}

\subsection{A compact representation of the SAS of a word}\label{sasRepresentation}
We now introduce a series of data structures which are fundamental for the efficient implementation of the main algorithms presented in this paper.

For a word $w$ of length $n$ with arch factorization $w = \ar_w(1) \cdots \ar_w(k) \rest(w)$,
we define two $k \times \sigma$ arrays $\firstPosArch[\cdot][\cdot]$ and $\lastPosArch[\cdot][\cdot]$ by the following relations.
For $\ell \in [1: k]$ and $a \in \Sigma$,
$\firstPosArch[\ell][a]$ is the leftmost position of $\ar_w(\ell)$ where $a$ occurs and
$\lastPosArch[\ell][a]$ is the rightmost position of $\ar_w(\ell)$ where $a$ occurs.
These two arrays are very intuitive:
they simply store for each letter of the alphabet its first and last occurrence inside the arch. To avoid confusions, for all $\ell$ and $a$, the values stored in $\firstPosArch[\ell][a]$ and $\lastPosArch[\ell][a]$  are positions of $w$ (so, between $1$ and $|w|$), they are not defined as positions of $\ar_w(\ell)$.

\begin{example}
For $w = 12213.113312.21$, with arches $\ar_w(1) = 12213$ and $\ar_w(2) = 113312$ and rest $\rest(w) = 21$,
we have the following.

\vspace{0.8em}
\bgroup
\def\arraystretch{1.2}%
\begin{tabular}{c|c|c|c}
	& $\ \firstPosArch[\cdot][1]\ $\ & $\ \firstPosArch[\cdot][2]\ $ & $\ \firstPosArch[\cdot][3]\ $ \\
	\hline
	$\ \firstPosArch[1][\cdot]\ $ & 1 & 2 & 5 \\
	$\ \firstPosArch[2][\cdot]\ $ & 6 & 11 & 8 \vspace{1em} \\

	& $\ \lastPosArch[\cdot][1]\ $ & $\ \lastPosArch[\cdot][2]\ $ & $\ \lastPosArch[\cdot][3]\ $ \\
	\hline
	$\ \lastPosArch[1][\cdot]\ $ & 4 & 3 & 5 \\
	$\ \lastPosArch[2][\cdot]\ $ & 10 & 11 & 9
\end{tabular}
\egroup
\end{example}
\vspace{-7pt}
\begin{lemma}\label[lemma]{lem:firstInArch}
	For a word $w$ of length $n$,
	we can compute $\firstPosArch[\cdot][\cdot]$ and $\lastPosArch[\cdot][\cdot]$ in $O(n)$ time.
\end{lemma}

\begin{proof}
	As illustrated by \cref{alg:firstLastPosArch},
	we can traverse the word $w$ from left to right,
	and similar to the computation of the arch factorisation,
	we keep track of the first encounters of each letter in each arch.
	The entries for the $\lastPosArch[\cdot]$ array will then be updated continuously until the end of an arch is reached.
	
\SetKw{KwAnd}{and}
\begin{algorithm}[!htb]
	\SetAlgoLined
	\KwIn{Word with its arch factorization $w = \ar_w(1) \cdots \ar_w(k) \rest(w)$ with $|w| = n$,
	alphabet $\Sigma$ with $\len \Sigma = \sigma$}
	\KwResult{Arrays $\firstPosArch$ und $\lastPosArch$ for word $w$}
	
	define array $\firstPosArch[1:k][1:\sigma]$\;
	define array $\lastPosArch[1:k][1:\sigma]$\;
	define array $C$ of size $\sigma$\;
	current position $posW \gets 1$\;
	
	\For{$i=1$ \KwTo $k$}{
		reset all elements of $C$ to $0$\;
		\ForEach{letter $a$ in $\ar_w(i)$}{
			\If{$C[a]$ == $0$}{
				$C[a] \gets 1$\;
				$\firstPosArch[i][a] = posW$\;
				$\lastPosArch[i][a] = posW$\;
			} \Else {
				$\lastPosArch[i][a] = posW$\;
			}
			$posW \gets posW + 1$\;
		}
	}
	\BlankLine
	
	\Return $\firstPosArch$ \KwAnd $\lastPosArch$\;
	
	\caption{$\firstPosArch_w(a, i)$ and $\lastPosArch_w(a, i)$}
	\label{alg:firstLastPosArch}
\end{algorithm}

Note that \cref{alg:firstLastPosArch} can be easily combined with \cref{alg:archfact}
to compute the arch factorisation and the arrays $\firstPosArch[\cdot]$ and $\lastPosArch[\cdot]$ while traversing the word $w$ once from right to left.
\end{proof}

We continue with the array $\minArch[\cdot]$ with $n$-elements, where, for $i\in [1:n]$,
$\minArch[i]=j$ if and only if $j=\min\{g\mid \al(w[i:g])\!=\Sigma\}$. If $\{g\mid \al(w[i:g])=\Sigma\}\!=\emptyset$, then $\minArch[i]\!=\!\infty$. Intuitively,
$\minArch[i]$ is the end point of the shortest $1$-universal factor (i.e., arch) of $w$ starting~\mbox{at $i$.} 

\begin{example}
Consider the word $w=12213.113312.21$ with the arches $\ar_w(1)=12213$ and $\ar_w(2)=113312$ and the rest $\rest(w)=21$.
We have $\minArch[j]=5$ for $j\in[1:3]$, $\minArch[j]=11$ for $j\in [4:9]$, and $\minArch[j]=\infty$ for $j\in [10:13]$.
\end{example}

The array $\minArch[\cdot]$ can be computed efficiently.
\begin{lemma}\label[lemma]{lem:minimalarches}
For a word $w$ of length $n$, we can compute $\minArch[\cdot]$ in $O(n)$ time.
\end{lemma}

\begin{proof}
The algorithm starts with a preprocessing phase in which we define an array count$[\cdot]$ of size $|\Sigma|$ and initialize its components with $0$. We also initialize all elements of the the array $\minArch[\cdot]$ with $\infty$.

In the main phase of our algorithm,
we use a $2$ pointers strategy.
We start with both $p_1$ and $p_2$ on the first position of $w$.
Then, the algorithm consists in an outer loop which has two inner loops.
The outer loop is executed while $p_2 \leq n$.
We describe in the following one iteration of the outer loop.
We start by executing the first inner loop where,
in each iteration,
we increment repeatedly $p_2$ by $1$
(i.e., advance it one position to the right in the word $w$),
until the factor $w[p_1:p_2]$ contains all letters of the alphabet
(i.e., it is $1$-universal).
If $p_2$ becomes $n+1$,
we simply stop the execution of this loop.
Then, the second inner loop is executed only if $p_2 \leq n$.
In each iteration of this loop,
while the factor $w[p_1:p_2]$ contains all the letters of the alphabet,
we set $\minArch[p_1] \gets p_2$,
and we advance $p_1$ one position towards the right.
This second loop stops as soon as $w[p_1:p_2]$ does not contain all the letters of the alphabet anymore.
Then, we continue to the next iteration of the outer loop
with the current values of $p_1$ and $p_2$ as computed in this iteration,
but only if $p_2\leq n$.

\medskip
The pseudo-code of this algorithm is described in \cref{alg:minArch}.

\SetKw{KwAnd}{and}
\SetKwFunction{KwMinArch}{minArch}
\begin{algorithm}[!ht]
	\SetAlgoLined
	\KwIn{Word $w $ with $|w| = n$,
		alphabet $\Sigma$ with $\len \Sigma = \sigma$}
	\KwOut{Array $\minArch$ for word $w$}
	\BlankLine
	
	define array count of length $\sigma$ with elements initially set to $0$\;
	define array $\minArch$ of length $n$ with elements initially set to $\infty$\;
	alph $\gets 0$, $p_1 \gets 1$, $p_2 \gets 1$\;
	\BlankLine
	
	\While{$p_2 \leq n$}{
		\While{alph $\neq \sigma$ \KwAnd $p_2 \neq n+1$}{
			count$[w[p_2]] \gets$ count$[w[p_2]] + 1$\;
			\If{count$[w[p_2]]$ == $1$}{
				alph $\gets$ alph $+ 1$\;
			}
			\If{alph $< \sigma$}{
				$p_2 \gets p_2 + 1$\;
			}
		}
		\BlankLine
		
		\If{$p_2$ == $n + 1$}{
			\Return $\minArch$\;
		}
		\BlankLine
		
		\While{ alph == $\sigma$}{
			$\minArch[p_1] \gets p_2$\;
			count$[w[p_1]] \gets$ count$[w[p_1]] - 1$\;
			\If{count$[w[p_1]]$ == $0$}{
				alph $\gets$ alph $- 1$\;
			}
			$p_1 \gets p_1 + 1$\;
		}
		$p_2 \gets p_2 + 1$\;
	}
	\BlankLine
	
	\Return $\minArch$\;
	\caption{minArch($w$)}
	\label{alg:minArch}
\end{algorithm}

\medskip
The correctness can be shown as follows.

Firstly, it is clear that at the beginning of the execution of an iteration of the outer loop $w[p_1:p_2-1]$ is not $1$-universal. Thus, the same property holds for all of the prefixes and suffixes of $w[p_1:p_2-1]$.  Now, let $a$ be the value of $p_1$ at this point (at the beginning of the execution of an iteration of the outer loop). Then, in the first inner loop we identify the smallest value of $b=p_2$ such that $w[a:b]=w[p_1:p_2]$ is $1$-universal, and it is correct to set $\minArch[a]\gets b$. In the second inner loop we increment $p_1$ to identify all the suffixes $w[p_1:b]$ of $u$ which are $1$-universal. None of these strings has a proper prefix which is also $1$-universal, as this would imply that $\minArch[a]< b$, a contradiction. So, for all the values $p_1$ visited in this loop it is correct to set  $\minArch[p_1]\gets b$. Finally, once we have found the largest $p_1$ such that $w[p_1:b]$ is $1$-universal, we increment $p_1$ and $p_2$ (i.e., $p_2$ becomes $b+1$) and restart from line $5$. This is correct as: $\minArch[j]$ is correctly set for $j<p_1$ and $w[p_1:p_2-1]$ is not $1$-universal, so the same reasoning as above applies.

Now assume that there exists a position $i$ with $j=\min\{g\mid \al(w[i:g])=\Sigma\}$ and we did not set $\minArch[i]=j$. This means that when $p_2$ was set to be $j$ during the execution of our algorithm, we had $p_1>i$. In order for this to hold, our algorithm should have already set $\minArch[i]=j'<j$. As shown above, this means that $w[i:j']$ is $1$-universal, a contradiction with the definition of $j$.

Further, we note that if the algorithm returns $\minArch$ in line $12$, then it can only do so because $p_2$ was incremented in line $10$ (so alph $<\sigma$). This means, that no new arch was found starting on $p_1$ and ending on a position $\leq n$. Therefore, $\minArch$ is correctly filled in. There is another possibility for the algorithm to end: we set $p_2=n+1$ in line $19$, and the outer while-loop is not executed anymore. As explained above, we have that $\minArch[j]$ is correctly set for $j<p_1$, and there is no $1$-universal word contained in $w[p_1:n]$ (so $\minArch[j]$ is also correctly set for $j\geq p_1$).

Now, let us discuss the complexity of the algorithm. The algorithm includes a preprocessing phase, which takes $O(n)$ time.

Each of the pointers $p_1$ and $p_2$ visits each position of $w$ exactly once, and in each step of the algorithm we increment one of these pointers. This incrementation is based on the result returned by a check whether $w[p_1:p_2]$ contains all the letters of the alphabet or not. This check can be done in constant time by maintaining an array which counts how many copies of each letter of the alphabet exist in the factor $w[p_1:p_2]$ and, based on it, a counter keeping track of $|\al(w[p_1:p_2)]|$.

\medskip
So, overall, the complexity of this algorithm is linear.
\end{proof}
\eject

A very important consequence of \cref{lem:minimalarches} is that we can define a tree-structure of the arches (i.e., the $1$-universal factors occurring in a word). Recall that we have defined $\llo(w)= \min\{\lastpos_w(a, \linebreak n)\mid a\in \Sigma\}$. That is, $\llo(w)$ is the position of the leftmost of the last occurrences of the letters of $\Sigma$ in $w$, and can be computed in linear time by \cref{lem:llo}. It is immediate to see that the letter $w[\llo(w)]$ does not occur in $\rest(w)$.
\begin{definition}\label{def:archTree}
Let $w$ be a word of length $n$ over $\Sigma$. The arch-tree of $w$, denoted by ${\mathcal A}_w$, is a rooted labelled tree defined as follows:
\begin{itemize}
\item The set of nodes of the tree is $\{i\mid 0\leq i\leq n+1\}$, where the node $n+1$ is the root of the tree ${\mathcal A}_w$. The root node $n+1$ is labelled with the letter $w[\llo(w)]$, the node $i$ is labelled with the letter $w[i]$, for all $i\in [1:n]$, and the node $0$ is labelled with $\uparrow$.
\item For $i\in [1:n-1]$, we have two cases:
\begin{itemize}
\item if $\minArch[i+1]=j$ then node $i$ is a child of node $j$.
\item if $\minArch[i+1]=\infty$ then node $i$ is a child of the root $n+1$.
\end{itemize}
\item Node $n$ is a child of node $n+1$.
\end{itemize}
\end{definition}
\begin{example}\label[example]{ex:tree}For the word $w=12213.113312.21$, the root $14$ of ${\mathcal A}$ has the children $9, 10,11, 12, 13$ and is labelled with $3$, the node $11$ has the children $3,4,5,6,7,8$, and the node $5$ has the children $0,1,2$. So, the root $14$ and the nodes $11$ and $5$ are internal nodes, while all other nodes are leaves.
\end{example}

\vspace{-8pt}
\begin{theorem}\label[theorem]{thm:archTree}
For a word $w$ of length $n$, we can construct ${\mathcal A}_w$ in $O(n)$ time.
\end{theorem}

\begin{proof}
This is a straightforward consequence of \cref{lem:minimalarches}.
\end{proof}

Constructed in a straightforward way based on \cref{lem:minimalarches}, the arch-tree ${\mathcal A}_w$ encodes all the arches occurring in $w$. Now, we can identify an $\sas$ in $w[i:n]$ by simply listing (in the order from $i-1$ upwards to $n+1$) the labels of the nodes met on the path from $i-1$ to $n+1$, without that of node $i-1$. 

For a word $w$ and each $i\in [0:n-1]$, we define $\depth(i)$ as the number of edges on the shortest path from $i$ to the root of $A_w$. In this case, for $i\in [0:n]$, we have that $w[i+1:n]$ has universality index $\depth(i)-1$.

\begin{example} For the word $w=12213.113312.21$ (same as in \cref{ex:tree}), an $\sas$ in $w[5:13]$ contains the letters $w[11]=2$ and the label $3$ of $14$, so $23$ (corresponding to the path $4\rightarrow 11\rightarrow 14$). An $\sas$ in $w[2:13]$ contains, in order, the letters $w[5]=3, w[11]=2$ and the label $3$ of $14$, so $323$ (corresponding to the path $1\rightarrow 5\rightarrow 11\rightarrow 14$). We also have $\depth(4)=2$, so $\iota(w[5:13])=1$, while $\depth(1)=3$ and $\iota(w[2:13])=2$.
\end{example}

Enhancing the construction of  ${\mathcal A}_w$ from \cref{thm:archTree} with level ancestor data structures \cite{BenderFarachTCS} for this tree, we can now efficiently process {\em internal $\sas$ queries} for a given word $w$, i.e., we can efficiently retrieve a compact representation of an $\sas$ for each factor of $w$.  \looseness=-1
\begin{theorem}\label[theorem]{thm:rangeQueries}
For a word $w$ of length $n$ we can construct in $O(n)$ time data structures allowing us to answer in $O(1)$ time  queries {\em $\sasRange(i,j)$: ``return a representation of an $\sas$ of $w[i:j]$". }
\end{theorem}

\begin{proof}
We start with the preprocessing phase, in which we construct the data structures allowing us to answer the queries efficiently.

We first construct the tree ${\mathcal A}_w$ for the word $w$. For each node $i$ of $A_w$, we compute $\depth(i)$. This can be done in linear time in a standard way.

\medskip
Then we construct a solution for the {\em Level Ancestor Problem} for the rooted tree ${\mathcal A}_w$. This is defined as follows (see \cite{BenderFarachTCS}). For a rooted tree $T$, $\LA_T(u,d)=v$ where $v$ is an ancestor of $u$ and $\depth(v)=d$, if such a node exists, or $\uparrow$ otherwise. The {\em Level Ancestor Problem} can now be formulated.
\begin{itemize}
\item Preprocessing: A rooted tree $T$ with $N$ vertices. ($T=A_w$ in our case and $N=n+1$)
\item Querying: For a node $u$ in the rooted tree $T$, query $\levelAncestor_{T} (u,d)$ returns $\LA_{T} (u,d)$, if it exists, and false otherwise.
\end{itemize}
A simple and elegant solution for this problem which has $O(N)$ preprocessing time and $O(1)$ time for query-answering can be found in, e.g., \cite{BenderFarachTCS} (see also \cite{BenAmram} for a more involved discussion). So, for the tree $A_w$ we can compute in $O(n)$ time data structures allowing us to answer $\levelAncestor_{A_w}$ queries in $O(1)$ time.

This is the entire preprocessing we need in order to be able to answer $\sasRange$-queries in constant time.

\medskip
We will now explain how an $\sasRange$-query is answered.

Let us assume we have to answer $\sasRange(i,j)$, i.e., to return a representation of an $\sas$ of $w[i:j]$. The compact representation of an $\sas$ of $w[i:j]$ will consist in two nodes of the tree $A_w$ and in the following we explain both how to compute these two nodes, and what is their semantics (i.e., how an $\sas$ can be retrieved from them).

Assume that $\depth(i-1)=x$ and $\depth(j)=y$.
We retrieve the ancestor $t$ of the node $i-1$ which is at distance $x-y$ from this node. So $t$ is the ancestor of depth $y$ of node $i-1$, and can be retrieved as $\levelAncestor_{A_w}(i-1,y)$. We check whether $t> j$ (i.e., $w[i:j] $ is a prefix of $w[i:t]$). If $t> j$, then we set $t'$ to be the successor of $t$ on the path to $i-1$. If $t\leq j$, we set $t'=t$.

The answer to $\sasRange(i,j)$ is the pair of nodes $(i-1,t')$. This answer can be clearly computed in constant time after the preprocessing we have performed.

We claim that this pair of nodes is a compact representation of an $\sas$ of $w[i:j]$, and such an $\sas$ can be obtained as follows: we go in the tree $A_w$ from the node $i-1$ on the path towards node $t'$ and output, in order, the labels of the nodes we meet (except the label of node $i-1$). Then we output the label of the parent node of $t'$. This is an $\sas$ of $w[i:j]$.

\medskip
Let us now explain why the above claim holds.

We recall the following basic fact: if $u$ and $v$ are words, then $\iota(u)+\iota(v)\leq \iota(uv)\leq \iota(u)+\iota(v)+1$. That is, $\iota(u)\in \{\iota(uv)-\iota(v), \iota(uv)-\iota(v)-1\}.$

Now, from the fact that $\depth(i-1)=x$ and $\depth(j)=y$ we get that $\iota(w[i:n])=x-1$ and $\iota(w[j+1:n])=y-1$. Therefore, from the fact recalled above (with $u=w[i:j]$ and $v=w[j+1:n]$), we have that $(x-1)-(y-1)-1=x-y-1\leq \iota(w[i:j])\leq (x-1)-(y-1)=x-y$. Thus, we compute $w[i:t]$, the shortest factor of $w$ starting on position $i$ which is $x-y$ universal; clearly, $t=\levelAncestor_{A_w}(i-1,y)$. Now, there are two possibilities. In the first case, $\iota(w[i:j])=x-y-1$ and $j< t$. Then, for $t'$ the successor of $t$ on the path to $i-1$, we have that $w[i:t']$ is the shortest $(x-y-1)$-universal prefix of $w[i:j]$ and $w[t]$ is a letter which does not appear in $w[t'+1:j]$. In the second case, $\iota(w[i:j])=x-y$ and $j\geq t$. Note that we also have that $\iota(w[i:t])=x-y$, so $j$ cannot be the parent of $t$ (otherwise, we would have $\iota(w[i:j])=x-y+1$). Then, in this case, we can take $t'=t$ and we have that $w[i:t']=w[i:t]$ is the shortest $(x-y)$-universal prefix of $w[i:j]$ and, if we take $t''$ to be the parent of $t'$ (that is, the parent of $t$), then $w[t'']$ is a letter which does not appear in $w[t+1:j]$ (which is a strict prefix of the shortest $1$-universal factor $w[t+1:t'']$ of $w$ starting on position $t+1$). Note again that, in this case, since $\iota(w[i:j])= x-y$ and $\iota(w[i:t''])=x-y+1$, we have that $j< t''$.

In both cases, the label of the nodes found on the path from node $i-1$ towards node $t'$, without the label of $i-1$, form the subsequence $m(w[i:j])$, from \cref{archfact}. To this subsequence we add either $w[t]$ (if $t\neq t'$) or $w[t'']$ (otherwise), and obtain an absent subsequence of length $\iota(w[i:j])+1$ of $w[i:j]$. This concludes our proof.
\end{proof}

In order to obtain compact representations of all the $\sas$ and $\mas$ of a word $w$, we need to define a series of additional arrays.

For a word $w$ of length $n$, we define the array $\dist[\cdot]$ with $n$-elements, where, for $i\in [1:n]$,
$\dist[i]=\min\{|u|\mid u$ is an absent subsequence of minimal length of $w[i:n],$ starting with $w[i]\}$.
The intuition behind the array $\dist[\cdot]$, and the way we will use it, is the following. Assume that $w$ has $k$ arches $\ar_w(1),\ldots,\ar_w(k)$, and $i\in \ar_w(\ell)$. Then there exists an $\sas$ of $w$, denoted by $w[i_1]\ldots w[i_{k+1}]$, which contains position $i$ only if $\dist[i]=k+1 - (\ell-1)=k-\ell+2$. Indeed, an $\sas$ contains one position of every arch, so if $i$ is in that $\sas$, then it should be its $\ell^{th}$ position, and there should be a word $u$ starting on position $i$, with $|u|=k-\ell+2$, such that $u$ is not a subsequence of $w[i:n]$. Nevertheless, for all positions $j$ of $\ar_w(\ell)$ it holds that $\dist[j]\geq k-\ell+2$.

\begin{example}
For the word $w=12213.113312.21$, we have $\dist[i]=2$ for $i\in [9:13]$ (as exemplified by the word $w[i]3$), $\dist[i]=3$ for $i\in [3:8]$ (as exemplified by the word $w[i]23$), and $\dist[i]=4$ for $i\in [1:2]$ (as exemplified by the word $w[i]323$). Note that the single shortest absent sequence in this word is $323$.
\end{example}

\vspace{-7pt}
\begin{lemma}\label[lemma]{lem:distance}
For a word $w$ of length $n$, we can construct $\dist[\cdot]$ in $O(n)$ time.
\end{lemma}

\begin{proof}
It is clear that $\dist[n]=2$. So, from now on, we show how to compute $\dist[i]$ for $i<n$.
It is not hard to see that an absent subsequence of minimal length of $w[i:n],$ starting with $w[i]$, can be obtained by concatenating $w[i]$ to an $\sas$ of $w[i+1:n]$. The latter can be computed, e.g., using $\sasRange[i+1:n]$, and its length is $\depth(i)$. We can thus conclude that, $\dist[i]=\depth(i)+1$. So, as $A_w$ can be computed in linear time, and the depth of the nodes of this tree can also be computed in linear time, the conclusion follows.
\end{proof}

Next, we introduce two additional arrays $\sortedLast[\cdot][\cdot]$ and $\Leq[\cdot][\cdot]$ which are crucial in our representation of the shortest absent subsequences of a word.
Let $w$ be a word of length $n$, with arch factorization $w=\ar_w(1)\cdots \ar_w(k)\rest(w)$. For each $\ell\in[1:k-1]$, we define $L_\ell$ to be the set $\{\lastPosArch [\ell][a]\mid a\in \Sigma\} $ of the rightmost positions of any letter $a\in\Sigma$ in the $\ell^{th}$ arch of $w$. Next we filter out all those positions from $L_\ell$ which cannot be the first letter of an $\sas$ of $\ar_w(l+1)\cdots \ar_w(k)\rest(w)$ and store the remaining positions in the set $L'_\ell=L_\ell \setminus \{\lastPosArch [\ell][a]\mid \dist(\firstPosArch[\ell+1][a])>k-\ell+1\}$. Finally, we define the array $\sortedLast[\ell][\cdot]$, with $|L'_\ell|$ elements, which contains the elements of $L'_\ell$ sorted in ascending order. For $\ell=k$, we proceed in the same way, except that in the second step we filter out all positions which correspond to letters occurring in $r(w)$. We define $L_k$ to be the set $\{\lastPosArch [k][a]\mid a\in \Sigma\} $ and $L'_k=L_k \setminus \{\lastPosArch [k][a]\mid a\in \al(\rest(w))\}$. Then, once more, we define the array $\sortedLast[k][\cdot]$, with $|L'_k|$ elements, which contains the elements of $L'_k$ sorted in ascending order.

\medskip
Moreover, we define for each $\ell$ an array with $|\ar_w(\ell)|$ elements $\Leq[\ell][\cdot]$ where $\Leq[\ell][i] = \max(\{t\mid \sortedLast[\ell][t]\leq i\}\cup\{-\infty\})$. For simplicity, we assume that $\Leq[\ell][\cdot]$ is indexed by the positions of $\ar_w(\ell)$ (i.e., if $\ar_w(\ell)=w[x:y]$ then $\Leq[\ell][\cdot]$ is indexed by the numbers from $x$\linebreak to $y$).

\medskip
The intuition behind these arrays is the following. Assume that $i$ is a position in $\ar_w(\ell)$. Then, by the greedy strategy described in the proofs of \cref{thm:mas,lem:sas}, $i$ can be contained in an $\sas$ only if $i=\firstPosArch[\ell][w[i]]$. Moreover, we need to have that $\dist[i]=k-\ell+2$, as explained when the intuition behind the array $\dist[\cdot] $ was described. That is,
$\lastPosArch[l-1][w[i]]$ needs to be in $L'_{\ell-1}$. Finally, an $\sas$ going through $i$ can only continue with a letter $a$ of the arch $\ell+1$ such that $\dist[\firstPosArch[\ell+1][a]]=k-\ell+1$, which, moreover, occurs last time in the arch $\ell$ on a position $\leq i$ (otherwise, the $\sas$ would have two letters in arch $\ell$, a contradiction). So, $a$ is exactly one of the letters on the positions given by the elements of $\sortedLast[\ell][1:\Leq[\ell][i]]$.

\begin{example}For $w=12213.113312.21$ we have $\ar_w(1)=12213$, $\ar_w(2)=113312$, and $L_1=\{3,4,5\}$. From this set, we eliminate positions $\lastPosArch[1][1]=4$ and $\lastPosArch[1][3]=5$, since $\dist[\firstPosArch[2][1]]=\dist[6]=3>2-1+1$ and $\dist[\firstPosArch[2][3]]=\dist[8]=3>2-1+1$. So $L'_1=\{3\}$. Thus, $\sortedLast[1][\cdot]$ has only one element, namely $3$. Accordingly, $\Leq[1][i]=-\infty$ for $i\in [1:2]$ and $\Leq[1][3]=1$ (corresponding to the element $\sortedLast[1][1]=3$ of the array $\sortedLast[1][\cdot]$); moreover, $\Leq[1][4]=\Leq[1][5]=1$. We have $L_2=\{9,10,11\}$. From this set, we once more eliminate two positions $\lastPosArch[2][1]=10$ and $\lastPosArch[2][2]=11$, because $1,2\in \al(\rest(w))$. So $L'_2=\{9\}$. Thus, $\sortedLast[2][\cdot]$ has only one element, namely $9$. Accordingly, $\Leq[2][i]=-\infty$ for $i\in [6:8]$ and $\Leq[2][9]=\Leq[2][10]=\Leq[2][11]=1$ (corresponding to the element $\sortedLast[2][1]=9$ of the array $\sortedLast[2][\cdot]$).
\end{example}

Based on the arrays $\sortedLast[\ell][\cdot]$, for $\ell\leq k$, we define the corresponding arrays $\Lex[\ell][\cdot]$ with, respectively, $|L'_\ell|$ elements, where $\Lex[\ell][i]=t$ if and only if $w[t]=\min\{w[j]\mid j\in \sortedLast[\ell][1:i]\}$. In other words, $w[\,\Lex[\ell][i]\,]$ is the lexicographically smallest letter occurring on the positions of $w$ corresponding to the first $i$ elements of $\sortedLast[\ell][\cdot]$.
Finally, we collect in a list $\startSAS$ the elements $\firstPosArch[1][a]$, for $a\in \Sigma$, such that $\dist[\firstPosArch[1][a]]=k+1$. We assume $\startSAS$ is ordered such that $\firstPosArch[1][a]$ comes before $\firstPosArch[1][b]$ in this list if $a$ is lexicographically smaller than $b$. Moreover, let $\mathtt{init}$ be the element of $\startSAS$ such that $w[\mathtt{init}]$ is lexicographically smaller than $w[j]$, for all $j\in \startSAS \setminus \{\mathtt{init}\}$; that is $\mathtt{init}$ is the first element in $\startSAS$.

\begin{lemma}
For a word $w$ of length $n$, we can construct $\sortedLast[\cdot][\cdot]$, $\Leq[\cdot][\cdot]$, $\Lex[\cdot][\cdot]$, and the list $\startSAS$ in $O(n)$ time.
\end{lemma}

\begin{proof}
By \cref{thm:archfact-algo}, we produce the arch factorisation of $w$ in linear time: $w=\ar_w(1)\cdots\ar_w(k)r(w)$.
By \cref{lem:distance,lem:firstInArch}, we compute in linear time the arrays $\lastPosArch[\cdot][\cdot]$ and $\firstPosArch[\cdot][\cdot]$ as well as the array $\dist[\cdot]$ and the set of letters $\al (r(w))$ (represented in a way which allows to test membership to this set in $O(1)$ time, such as, for instance, the array $f$ from the proof of \cref{thm:archfact-algo}).

Then, we can obtain also in linear time the lists $L_\ell$ and $L'_\ell$, for all $\ell$, in a direct way: we simply iterate through all the arches, and for each arch $\ar_w(\ell)$ we iterate through the letters $a$ of $\Sigma$, and collect in $L_\ell'$ those positions $\lastPosArch[\ell][a]$ which fulfil the requirements in the definition of this set (which can be clearly tested now in $O(1)$ time).

Note that, when computed in this way, $L_\ell'$ is not sorted. We can, however, sort in linear time $O(n)$ the set $\cup_{\ell'\leq k}L_\ell'$ in increasing order by counting sort (as all values in these sets are between $1$ and $n$), and obtain an array $S$ of positions of $w$. We then produce the arrays $\sortedLast[\ell][\cdot]$ in order, for $\ell$ from $1$ to $k$, by selecting from $S$ the subarray containing positions which are contained in the arch $\ar_w(\ell)$. This is done as follows: we traverse $S$ left to right with a pointer $p$. After we have produced all the arrays $\sortedLast[i][\cdot]$ for $i<\ell$, $p$ points to the leftmost element $t$ of $S$ such that $t> |\ar_w(1)\cdots\ar_w(\ell-1)|$. We then copy in $\sortedLast[i][\cdot]$ all the elements $t'$ of $S$ which occur to the right of $p$ and fulfil $t'\leq |\ar_w(1)\cdots\ar_w(\ell)|$. We then move $p$ to the element of $S$ found immediately to the right of the last copied element. Clearly, this process takes linear time.

Computing each array $\Leq[\ell][\cdot]$ is done in a relatively similar way. Intuitively, it is also similar to the merging of the sequence of positions of $\ar_w(\ell)$ (considered in increasing order) to the array $\sortedLast[\ell][\cdot]$. We traverse the arch $\ar_w(\ell)$ left to right, keeping track of the largest element of the sorted array $\sortedLast[\ell][\cdot]$ which we have met so far in this traversal.  When we are at position $i$ in the arch and $\sortedLast[\ell][j]$ is the largest element of the sorted array $\sortedLast[\ell][\cdot]$ which we have met when considering the positions up to $i-1$, we check if $i=\sortedLast[\ell][j+1]$. If yes, we set $\Leq[\ell][i]=j+1$ and  $\sortedLast[\ell][j+1]$ becomes the largest element of $\sortedLast[\ell][\cdot]$ seen so far; otherwise we set  $\Leq[\ell][i]=j$. Clearly, the computation of a single array $\Leq[\ell][\cdot]$ can be implemented in $O(|\ar_w(\ell)|)$ time, so computing all these arrays takes linear time.

Finally, $\Lex[\ell][\cdot]$ can be computed in $O(|L_\ell'|)$ time, for each $\ell\leq k$, by a simple left to right traversal of the array $\sortedLast[\ell][\cdot]$. So, the overall time needed to compute these arrays is linear. The same holds for the computation of $\startSAS$. Clearly, this list can be constructed in $O(\sigma)$ time by simply inserting in it the elements $\firstPosArch[1][a]$, for $a\in \Sigma$, such that $\dist[\firstPosArch[1][a]]=k+1$, while going through the letters $a\in \Sigma$ in lexicographical order.

\medskip
This concludes the proof of our statement.
\end{proof}

We now formalize the intuition that
$\firstPosArch[1][\cdot]$, $\sortedLast[\cdot][\cdot]$, and $\Leq[\cdot][\cdot]$
are a compact representation of all the $\sas$ of a given word $w$.
\begin{theorem}\label[theorem]{thm:sasRep}
Given a word $w$ of length $n$ with $k$ arches $\ar_w(1),\ldots,\ar_w(k)$, as well as the arrays $\firstPosArch[\cdot][\cdot]$, $\sortedLast[\cdot][\cdot]$, $\Leq[\cdot][\cdot]$, $\Lex[\cdot][\cdot]$, the set $\startSAS$, and the position $\mathtt{init}$ of $w$, which can all be computed in linear time, we can perform the following tasks:
\begin{enumerate}
\item We can check in $O(k)$ time if a word $u$ of length $k+1$ is an $\sas$ of $w$.
\item We can compute in $O(k)$ time the lexicographically smallest $\sas$ of $w$.
\item We can efficiently enumerate (i.e., with polynomial delay) all the $\sas$ of $w$.
\end{enumerate}
\end{theorem}
\begin{proof}
Following \cref{archfact}, let $w=\ar_w(1) \cdots \ar_w(k) r(w)$.

\medskip
We observe that the arrays $\firstPosArch[\cdot][\cdot]$, $\sortedLast[\cdot][\cdot]$, $\Leq[\cdot][\cdot]$, and the set $\startSAS$ induce a tree structure ${\mathcal T}_w$ for $w$, called the $\sas$-tree of $w$, and defined inductively as follows:
\begin{itemize}
\item The root of ${\mathcal T}_w$ is $\bullet $, while all the nodes of ${\mathcal T}_w$ of depth $1$ to $k$ are positions of $w$. The nodes of depth $k+1$ of ${\mathcal T}_w$ are letters of $\Sigma$ and the children of each node are totally ordered.
\item The children of $\bullet$ are the positions of the list $\startSAS$. These are the only nodes of depth $1$. Moreover, the children of $\bullet$ are ordered in the same order as the corresponding positions of $\startSAS$.
\item For $\ell\in [1:k-1]$, we define the nodes of depth $\ell+1$ as follows. The children of a node $i$ of depth $\ell$ are the nodes $\firstPosArch[\ell+1][w[j]]$ where $j\in \sortedLast[\ell][1:\Leq[\ell][i]]$. The children of node $i$ are ordered such that the node corresponding to $\firstPosArch[\ell+1][w[j]]$ comes before $\firstPosArch[\ell+1][w[j']]$ if and only if $j$ comes before $j'$ in $\sortedLast[\ell][1:\Leq[\ell][i]]$.
\item For $\ell=k$, we define the nodes of depth $k+1$ (the leaves) as follows. The children of a node $i$ from level $k$ are the letters of $\{w[j]\mid j \in \sortedLast[\ell][1:\Leq[\ell][i]]\}$. The children of $i$ are ordered such that the node corresponding to letter $w[j]$ comes before the node corresponding to node $w[j']$  if and only if $j$ comes before $j'$ in $\sortedLast[\ell][1:\Leq[\ell][i]$.
\item ${\mathcal T}_w$ contains no other nodes than the ones defined above.
\end{itemize}

We claim that the paths in the tree from $\bullet$ to the leaves correspond exactly to the set of $\sas$ of $w$.

\medskip
Let us show first that each $\sas$ $u$ corresponds to a path in the tree ${\mathcal T}_w$ from $\bullet$ to a leaf. From the definition of $\sas(w)$, we have that $|u|=k+1$.
Firstly, it is clear from the definition of $\dist[\cdot]$ and $\startSAS$, that each $\sas$ starts with a position in $\startSAS$. Then, for $i\in [2:k]$, assume that the letter $u[i-1]$ occurs on position $j\in \ar_w(i-1)$ of $w$, and $j$ is a node of the tree.  Then, $\lastPosArch[i-1][u[i]]\leq j$ (or $\nextpos_w(u[i],j)$ would also be in $\ar_w(i-1)$), so $\lastPosArch[i-1][u[i]]\in \sortedLast[i-1][1:\Leq[i-1][j]]$. Thus $\firstPosArch[i][u[i]]$ is a child of $j$. Applying this inductive argument, we get that $u[1:k]$ is a path in ${\mathcal T}_w$ from $\bullet$ to a node of depth $k$. Assume now that the letter $u[k]$ occurs on position $j\in \ar_w(k)$ of $w$, so, as we have seen, $j$ is a node of the tree.  Then, $\lastPosArch[k][u[k+1]]\leq j$ and $u[k+1]\in \Sigma\setminus \al(r(w))$. Thus, $u[k+1]\in \{w[t]\mid t \in \sortedLast[k][1:\Leq[k][j]]\}$. This means that $u[k+1]$ is a child of $j$.

To show that each path of the tree from $\bullet$ to the leaves encodes an $\sas$ of $w$ is based on the following inductive argument. Firstly, as an induction base, we consider paths containing only two nodes, namely  $\bullet, i_1$; if $i_1\in \startSAS$, then $u[1]=w[i_1]$ is a word of length $1$, whose first occurrence in $w$ is on position $p=i_1$ of $w$ such that $\dist[p]=k+1$. Now, we define the induction hypothesis. Consider $\ell>1$ and let us assume that if we have a path $\bullet, i_1,i_2,\ldots,i_\ell$ in ${\mathcal T}_w$, then $i_\ell\in \ar_w(\ell)$, $\dist[i_\ell]=k-\ell+2$, and $u[1:\ell]=w[i_1]\cdots w[i_\ell]$ is a subsequence of $w$ such that the shortest factor of $w$ which contains $u[1:\ell]$ is $w[1:i_\ell]$. Now, in the induction step, we need to show that the same property holds for a path starting with the root $\bullet$ and containing $\ell+1$ other nodes. So, basically, we consider a path $\bullet, i_1,i_2,\ldots,i_\ell, i_{\ell+1}$ in ${\mathcal T}_w$. At this point we note that, in order to extend a path $\bullet, i_1,i_2,\ldots,i_\ell$ with a node $i_{\ell+1}$, we need to pick $i_{\ell+1}$ to be a child of $i_\ell$, and the path $\bullet,i_1,\ldots,i_\ell$ fulfils the induction hypothesis. Therefore, $i_\ell\in \ar_w(\ell)$, $\dist[i_\ell]=k-\ell+2$, and $u[1:\ell]=w[i_1]\cdots w[i_\ell]$ is a subsequence of $w$ such that the shortest factor of $w$ which contains $u[1:\ell]$ is $w[1:i_\ell]$. Moreover, as said, $i_{\ell+1}$ is a child of the node $i_\ell$, which means that there exists some $j$ such that $i_{\ell+1}=\firstPosArch[\ell+1][w[j]]$ where $j\in \sortedLast[\ell][1:\Leq[\ell][i_{\ell}]]$. Therefore, we get that $i_{\ell+1}\in \ar_w(\ell+1)$ and $\dist[i_{\ell+1}]=k-\ell+1$, Moreover, we also obtain that $u[1:\ell+1]=w[i_1]\cdots w[i_\ell]w[i_{\ell+1}]$ is a subsequence of $w$ such that the shortest factor of $w$ which contains $u[1:\ell+1]$ is $w[1:i_{\ell+1}]$ (as $i_{\ell+1}$ is the first position after $i_\ell$ where $w[i_{\ell+1}]$ occurs. This proves the induction step.

\medskip
In conclusion, we can apply this inductive argument for $\ell=k$, so if we have a path $\bullet, i_1,i_2,\ldots,i_k$ in ${\mathcal T}_w$, then $i_k\in \ar_w(k)$, $\dist[i_k]=2$, and $u[1:k]=w[i_1]\cdots w[i_k]$ is a subsequence of $w$ such that the shortest factor of $w$ which contains $u[1:k]$ is $w[1:i_k]$. Now, this path can only be extended by a child of $i_k$. The children of $i_k$ are the letters $a$ of $\{w[j]\mid j \in \sortedLast[k][1:\Leq[k][i_k]]\}$. It is thus clear that $u[1:k]a$ is an absent factor of $w$ of length $k+1$, thus an $\sas$.

\medskip
Based on this observation, we can prove the statement of the theorem. Note that we do not need to construct this (potentially very large) tree: it is compactly represented by the data structures enumerated in the statement.

\medskip
For (1) we simply have to check whether $u$ is encoded by a path in ${\mathcal T}_w$. As we do not actually construct this tree, we use the following approach. First we compare $m=|u|$ to $k+1$. If $m\neq k+1$, then $u$ is not an $\sas$ of $w$. If $m=k+1$, we continue as follows. For $i$ from $1$ to $m-1$, we check if $\dist[\firstPosArch[i][u[i]]]=k-i+2$ and $\firstPosArch[i][u[i]]\geq \lastPosArch[i][u[i+1]]$ (i.e., if $u[1:k]$ is a path in the tree). Then, if all these checks returned true, we check if $u[k+1]$ is in $ \{w[j]\mid j \in \sortedLast[k][1:\Leq[k][t]]\}$, where $t=\firstPosArch[k][u[k]]$. If this final check returns true, then $u$ is an $\sas$ of $w$. Otherwise, $u$ is not an $\sas$.

\medskip
The correctness follows from the properties of ${\mathcal T}_w$.

\medskip
For (2), we compute an $\sas$ $u$ of length $k+1$. We define $u[1]=w[\mathtt{init}]$. For $i$ from $2$ to $k+1$, we define the $i^{th}$ letter of $u$ as follows:
$$u[i]=w[\,\,\Lex[\,i-1\,][\,|L'_{i-1}|\, ]\,\,], \mbox{ where }|L'_{i-1}|\mbox{ is the size of }\Lex[i-1][\cdot] .$$

Clearly, this algorithm is derived from the argument showing how the paths of ${\mathcal T}_w$ correspond to the $\sas$ of $w$. When constructing an $\sas$ we can nondeterministically choose a path in ${\mathcal T}_w$. To construct the lexicographically smallest $\sas$, we make the choices in the construction of the $\sas$ in a deterministic way. Our algorithm clearly produces an $\sas$ of $w$. Moreover, at each step, when we have multiple choices to extend the word $u$ by a letter, we greedily choose this letter to be the lexicographically smallest from the possible choices, which clearly leads to the construction of the lexicographically smallest $\sas$ of $w$.

\medskip
For (3), we actually need to traverse the tree ${\mathcal T}_w$ and output the absent subsequences encoded by the paths that lead from $\bullet$ to a leaf, with polynomial delay. If ${\mathcal T}_w$ would be explicitly constructed, this would be done using a standard depth-first traversal of the tree ${\mathcal T}_w$, implemented in a recursive manner: each time a leaf is reached, we need to output the word whose letters correspond to the nodes stored in the recursive-calls stack, ending with the letter corresponding to the leaf. In our case, ${\mathcal T}_w$ is not stored explicitly. However, this is not a problem. We still do a recursive depth-first search in the tree ${\mathcal T}_w$ starting in its root $\bullet$. Each time we call the depth-first search for a node, we first produce the ordered list of its children (this can be clearly done in polynomial time, using the definition of the tree), and then we recursively call the depth-first search for each node in this list. When the list is fully processed, we end the respective call. When we call the depth-first search procedure for a leaf (in the following, we will call this a leaf-call, for simplicity), we output the word of length $k+1$ which corresponds to the contents of the recursive-calls stack, or, in other words, to the path leading from the root to that leaf (and this encodes, indeed, to a $\sas$).

\medskip
To compute this algorithm's time complexity, we need to make some observations. Firstly, we note that in each recursive call of the depth-first search, we do at most $\sigma$ steps, in order to produce the ordered list of children of that node. Then, we do a recursive call for each of them. That is, we use $O(\sigma)$ time to produce the children list, and then do $O(\sigma)$ new calls. Secondly, let us analyse what happens between two consecutive leaf-calls of the search. In the first leaf-call, we first output the word corresponding to the content of the stack (this takes $O(k)$ time) and then return. Once this leaf-call is ended, a series of other previous calls from the recursive-call stack might also end (the ones corresponding to final nodes in list of childrens); clearly, there are at most $k$ such recursive calls that could end. Then, a call is reached, which was made for a node whose list of children was not yet fully processed. A new call is now initiated, for the first unexplored node of that list, and this leads to a new sequence of at most $k$ recursive-calls ending with a leaf-call. So, to do the processing between two consecutive leaf-calls, we need at most $O(k\sigma)$ time. So: our algorithm simulates a depth-first search of the tree ${\mathcal T}_w$ and outputs all its paths with constant delay. By our argument that there is a one-to-one correspondence between the paths of ${\mathcal T}_w$ and the words from $\sas(w)$, the correctness of the algorithm is clear. Moreover, as the time between two words are output is upper bounded by $O(k\sigma)$, our claim follows.
\end{proof}

The main point of the previous theorem is to define a compact representation of the words of the set $\sas(w)$. The fact that this set can be, in fact, represented as a tree of depth $O(k)$ and branching factor $O(\sigma)$ allows us to immediately show that a series of algorithmic tasks can be performed efficiently. However, we do not claim that the solution proposed above for the enumeration of the words in $\sas(w)$ is optimal. It remains an open problem to define faster enumeration algorithms for the elements of this set.

\subsection{A compact representation of the MAS of a word}\label{masRepresentation}

Similar to the previous section we construct a data structure to store, for given $w\in\Sigma^\ast$, all the elements of $\mas(w)$. Since \ref{sasRepresentation} does not hold for $\mas$ this data structure is larger than the one for $\sas$.
We start by defining an array $\nextArray[\cdot]$ of size $n$ holds,
for each given position $i$,
the first position of $w[i]$ in the word $w[i+1:n]$.
We formally define this by
$\nextArray[i] = \nextpos(w[i],i+1) = \min \lbrace j \mid w[j] = w[i], j > i \rbrace$,
while we assume $\nextArray[i] = \infty$
if $w[i] \notin \al(w[i+1:n])$.

\medskip
The array $\prevArray[\cdot]$ of size $n$ holds,
for a given position $i$,
the last position of $w[i]$ in the word $w[1:i-1]$.
We formally define this with
$\prevArray[i] = \lastpos(w[i],i-1) = \max \lbrace j \mid w[j] = w[i], j < i \rbrace$,
while we assume $\prevArray[i] = -\infty$
if $w[i] \notin \al(w[1:i-1])$. Using \cref{alg:next} (\cref{alg:last}) to compute $\nextArray[\cdot]$ ($\prevArray[\cdot]$) is not efficient. Nevertheless we can compute both arrays in linear time.

\begin{lemma}\label[lemma]{lem:arrayX}
	Given a word $w$ with length $\len w = n$, we can compute in $O(n)$ time the arrays
	$\nextArray[\cdot]$ and $\prevArray[\cdot]$.
\end{lemma}
\begin{proof}
	The word $w$ needs to be traversed only once from left to right and once from right to left while maintaining an array of size $\len \Sigma$ with the last seen occurrence of each character.
	Since $\len \Sigma \leq n$, the results follow immediately.
\end{proof}

The following theorem is based on well-known search algorithms on (directed acyclic) graphs, see e.g. \cite{Sedgewick}.

\begin{theorem}\label[theorem]{thm:masRep2}
For a word $w$, we can construct in $O(n^2\sigma)$ time data structures allowing us to efficiently perform the following tasks:
\begin{enumerate}
\itemsep=0.8pt
\item We can check in $O(m)$ time if a word $u$ of length $m$ is an $\mas$ of $w$.
\item We can compute in polynomial time the longest $\mas$ of $w$.
\item We can check in polynomial time for a given length $\ell$ if there exists an $\mas$ of length $\ell$ of $w$.
\item We can efficiently enumerate (with polynomial delay) all the $\mas$ of $w$.
\end{enumerate}
\end{theorem}

\begin{proof}
For a word $w$, the compact representation of $\mas(w)$ is based on \cref{thm:mas}.
We define a directed acyclic graph ${\mathcal D}_w$ with the nodes $\{(i,j)\mid 0\leq j< i\leq n\}\cup \{(0,0),f\}$. The edges (represented as arrows $A \rightarrow B$ between nodes $A$ and $B$) are defined as follows:
\begin{itemize}
\itemsep=0.85pt
\item We have an edge $(0,0)\rightarrow (i,0)$, if there exists $a\in \Sigma$ such that $i=\firstPosArch[1][a]$. This edge is labelled with $a$.
\item For $1\leq j< i< k\leq n$ we have an edge $(i,j)\rightarrow (k,i)$,  if there exists $a\in \Sigma$ such that $k=\nextpos_w(a,i+1)$ and $a\in \al(w[j+1:i])$. This edge is labelled with $a$.
\item We have an edge $(i,j)\rightarrow f$, if there exists $b\in \al(w[j+1:i])$ and $b\notin  \al (w[i+1:n])$. This edge is labelled with $b$.
\end{itemize}

We claim that the words in $\mas(w)$ correspond exactly to the paths in the graph ${\mathcal D}_w$ from $(0,0)$ to $f$.

\medskip
Let $v$ be an $\mas$ of $w$, with $|v|=m$. By \cref{thm:mas}, there exist positions $0=i_0<i_1<\ldots <i_m <i_{m+1}= n+1$ such that all of the following conditions are satisfied:
\begin{enumerate}[label=(\roman*)]
\itemsep=0.95pt
	\item $v=w[i_1]\cdots w[i_m]v[m+1]$
	\item $v[1]\notin \al(w[1:i_1-1])$
	\item $v[k]\notin\al(w[i_{k-1}+1:i_k-1])$ for all $k\in[2: m+1]$
	\item $v[k]\in \al(w[i_{k-2}+1:i_{k-1}])$ for all $k\in[2: m+1]$
\end{enumerate}
It immediately follows that $i_1=\firstPosArch[1][v[1]]$. So, there is an edge $(0,0)\rightarrow (i_1,0)$. Then, clearly, for all $m\geq \ell\geq 2$, $i_\ell=\nextpos_w(v[\ell],i_{\ell-1}+1)$. Moreover, $v[\ell]\in \al(w[i_{\ell-2}+1:i_{\ell-1}])$. Thus, there is an edge $(i_{\ell-1},i_{\ell-2})\rightarrow (i_\ell,i_{\ell-1})$. Similarly, there is an edge from $(i_m,i_{m-1})$ to $f$.

\medskip
For the reverse implication, we consider a path  in the graph ${\mathcal D}_w$ from $(0,0)$ to $f$, namely $(0,0)$, $(i_1,0)$, $(i_2,i_1), \ldots, (i_m,i_{m-1}), f$.  We claim that the word $v=w[i_1]\cdots w[i_m] b$ is an $\mas $ of $w$, for some $b$ such that $b\notin \al w[i_m+1:n]$ and $b\in \al(w[i_{m-1}+1:i_m])$. It is immediate that $v$ fulfils the four conditions in the statement of \cref{thm:mas}.

\medskip
Before solving tasks $(1-4)$, we note that ${\mathcal D}_w$ can be constructed in $O(n^2\sigma)$ time, in a straightforward manner. We store this graph by having for each of its nodes a list of incoming edges (note, at this point, that all these edges have the same label). Using \cref{alg:next} trivially, we can produce in $O(n^2\sigma)$ an oracle data structure which contains the answers to any query $\nextpos_w(a,i)$, for all $a\in \Sigma$ and $i\leq n$ (note that a more efficient computation of such a data structure is described in \cref{cor:masExtension}). Firstly, we define an edge $(0,0)\rightarrow (i,0)$ for all $i$ such that $i=\firstPosArch[1][a]$ for some $a\in \Sigma$. Then, for all $i\leq n$ and $a\in \Sigma$, we retrieve $k=\nextpos_w(a,i+1)$ and, then, for all $j<i$, we check if $a\in \al(w[j+1:i])$ (by simply using our $\nextpos_w$-data structure to see if $\nextpos_w(a,j+1)\leq i$); if $a\in \al(w[j+1:i])$, we have an edge from $(i,j)$ to $(k,i)$ labelled with $a$. Finally, for all $j<i$ and $b\in \Sigma$, we have an edge $(i,j)\rightarrow f$ labelled with $b$, if $b\in \al(w[j+1:i])$ (i.e., $\nextpos_w(b,j+1)\leq i$) and $b\notin  \al (w[i+1:n])$ (i.e., $\nextpos_w(b,i+1)\leq n$ does not hold).

Once ${\mathcal D}_w$ constructed, the correspondence between the words of $\mas(w)$ and the paths from $(0,0)$ to $f$ in this graph gives us a way to solve all the problems mentioned in the statement of the theorem efficiently. For (2), we simply need to identify the longest such path in ${\mathcal D}_w$, which can be solved in polynomial time in the size of the graph, by a folklore algorithm (see, e.g., \cite{Sedgewick}). Then, for (3) we need to check if there exists such a path of length $\ell$; this can also be done trivially in polynomial time in the size of ${\mathcal D}_w$, by maintaining, for $h=1$ to $\ell$, the set of nodes at distance $h$ from $(0.0)$. Finally, for (4), we need to enumerate efficiently all paths in the graph ${\mathcal D}_w$ from $(0,0)$ to $f$. This can be implemented in polynomial time, see \cite{Wasa16} and the references therein.

\medskip
We only explain how (1) is performed. Firstly, we can store ${\mathcal D}_w$ as a  $(n+1)\times (n+1)\times \sigma$ three dimensional array $D[\cdot][\cdot][\cdot]$, where $D[i][j][a]=k$, for $i,j\in [0:n]$ and $a\in \Sigma$, if and only if there exists an edge from $(i,j)$ to $(k,i)$, and $w[k]=a$. In a sense, in this way we can see the graph ${\mathcal D}_w$ as a DAG whose edges are labelled with letters of $\Sigma$. Now, to solve (1), we need to check whether there exists a path in ${\mathcal D}_w$ from $(0,0)$ to $f$, labelled with the word $u$. This can be clearly done in $O(m)$ time.
\end{proof}

Note that the main point of the theorem above is to introduce the graph ${\mathcal D}_w$ as a compact representation of $\mas(w)$, and to show the correspondence between the strings in $\mas(w)$ and the paths from $(0,0)$ to $f$ in this graph. Once this correspondence established, we were only interested in establishing that the tasks $(2-4)$ can be solved in polynomial time, and this follows directly from general results regarding directed acyclic graphs. We expect that obtaining more efficient algorithms for the respective problems for the tasks $(2-4)$ from the theorem above would require a deeper understanding of the structure of ${\mathcal D}_w$, and leave this as an open problem.

The last result of this paper shows how to preprocess a word $w$ in order to be able to answer efficiently queries of the following form: for a given $u$, that is a subsequence of $w$, which is the shortest string we can append to $u$ in order to obtain an $\mas$ of $w$?

\begin{corollary}\label[corollary]{cor:masExtension}
For a word $w$ of length $n$, we can construct in $O(n\sigma)$ time data structures allowing us to answer $\masExtend(u)$ queries: for a subsequence $u$ of $w$, decide whether there exists an $\mas$ $uv$ of $w$, and, if yes, construct such an $\mas$ $uv$ of minimal length. The time needed to answer a query is $O(|v|+|u|)$.
\end{corollary}

\begin{proof}
Before we begin to show this theorem, we recall an important data structure allowing us to answer range maximum queries.
Let $A$ be an array with $n$ elements from a well-ordered set. We define {\em range maximum queries} $\RMQ_A$ for the array of $A$:
$\RMQ_A(i,j)=\argmax \{A[t]\mid t\in [i:j]\}$, for $i,j\in [1:n]$. That is, $\RMQ_A(i,j)$ is the position of the largest element in the subarray $A[i:j]$; if there are multiple positions containing this largest element, $\RMQ_A(i,j)$ is the leftmost of them. (When it is clear from the context, we drop the subscript $A$).

We will use the following result from \cite{rmq}.
\begin{lemma}\label[lemma]{RMQ}
Let $A$ be an array with $n$ integer elements. One can preprocess $A$ in $O(n)$ time and produce data structures allowing to answer in constant time {\em range maximum queries} $\RMQ_A(i,j)$, for any $i,j\in [1:n]$.
\end{lemma}

We now start the actual proof of our theorem and first describe the preprocessing phase.

\medskip
We first use \cref{lem:arrayX} to compute the arrays $\nextArray$ and $\prevArray$.

We define an $n\times \sigma$ matrix $X[\cdot][\cdot]$, where $X[i][a]=\nextpos_w(a,i+1)$, for all $i\in [0:n-1]$ and $a\in \Sigma$. This matrix can be computed in $O(n\sigma)$ time as follows. We go with $i$ from $1$ to $n$, and we set $X[j][w[i]]=i$ for all $j \in [\prevArray[i]:i-1]\cap[1:n]$. Clearly we can compute an array of size $n$ which contains, for all $i\in [1:n]$, the interval $[\prevArray[i]:i-1]\cap[1:n]$ in linear time.

We define an $n\times \sigma$ matrix $Y[\cdot][\cdot]$, where $Y[i][a]=\lastpos_w(a,i-1)$, for all $i\in [2:n+1]$ and $a\in \Sigma$. This matrix can be computed in $O(n\sigma)$ time trivially, as it equals to the matrix $X$ computed for the mirror image of the word $w$.

We also preprocess the array $\nextArray[\cdot]$ as in \cref{RMQ}, so that we can answer in constant time $\RMQ_{\nextArray}$-queries.

To answer a query $\masExtend(u)$, we proceed as follows.
If $|u|=0$, we choose a letter $a$ not occurring in $r(w)$ and simply return $m(w)a$ (where $m(w)$ denotes the concatenation of the last letters of the arches of $w$) (see the proof of \cref{lem:oneSAS}), which is clearly an $\sas$ of $w$.
If $|u|>0$, we run \cref{alg:isMASubseq,alg:masExtend}, and return the word $v$ computed by the latter.

The correctness of the approach above is explained in the following. For the rest of this proof, assume thus that $|u|>0$ and that $u$ is a subsequence of $w$.

\medskip
First, we run \cref{alg:isMASubseq}, which decides whether $u$ can be the prefix of an $\mas$ of $w$ or not. This is an extended and modified version of \cref{alg:isSubseq}, in which we also check the conditions from \cref{thm:mas}. Note again that we apply this algorithm under the assumption that $u$ is a subsequence of $w$.

\SetKw{KwAnd}{and}
\begin{algorithm}[!htb]
	\SetAlgoLined
	\KwIn{Word $w$ with $|w| = n$, word $u$ with $|u|=m$ to be tested}
	
	$j_0\gets 0; j_1 \gets X[0][u[1]]; i\gets 2;$ flag$\gets $ true\;
	\While{$i\leq m$ \KwAnd flag == true }{
	 	$j_2 \gets X[j_1][u[i]]$\;
	 	$j_3 \gets Y[j_2][u[i]]$\;
		\If{$j_3\notin [j_0+1:j_1]$}{
			flag $\gets $ false\;
		} \Else {
			$j_0\gets j_1; j_1\gets j_2; i\gets i+1 $\;
		}
	}
	
	\Return flag == true ? $(j_0,j_1)$ : $\uparrow$\;
	
	\caption{isMASPrefix($w$,$u$)}
	\label{alg:isMASubseq}
\end{algorithm}

If \cref{alg:isMASubseq} returns $\uparrow$, then $u$ cannot be the prefix of an $\mas$ of $w$, as it does not fulfil the conditions of \cref{thm:mas}.

\medskip
If \cref{alg:isMASubseq} returns a pair $(j_0,j_1)$, with $j_0<j_1$, then $u$ can be extended to an $\mas$ of $w$, $w[1:j_0]$ is the shortest prefix of $w$ which contains $u[1:m-1]$ and $w[1:j_1]$ is the shortest prefix of $w$ which contains $u[1:m]$. In particular, $u[m-1]=w[j_0]$ and $u[m]=w[j_1]$. We continue with the identification of the shortest string $v$ which we can append to $u$, based on \cref{thm:mas}, so that $uv$ becomes an $\mas$ of $w$.  We achieve this by \cref{alg:masExtend}.

\begin{algorithm}[!htb]
	\SetAlgoLined
	\KwIn{Word $w$ with $|w| = n$, two positions $j_0<j_1$ of $w$}
	
	$i_0\gets j_0; i_1 \gets j_1; i\gets 1;$ flag$\gets $ true\;
	\While{ flag == true }{
	 	$i_2 \gets \RMQ_{\nextArray}(i_0+1,i_1)$\;
		\If{$\nextArray[i_2]>n $} {
			$v[i]\gets w[i_2];$ flag $\gets $ false\;
		} \Else {
			$v[i]\gets w[i_2]; i_0\gets i_1; i_1\gets \nextArray[i_2]; i\gets i+1 $\;
		}
	}
	
	\Return $v$\;
	
	\caption{masExtend($w,j_0,j_1$)}
	\label{alg:masExtend}
\end{algorithm}

The idea of the algorithm is the following. In order to extend $u$ to a $\mas$ of $w$, by \cref{thm:mas}, we need to find a position $j$ between $i_0+1$ and $i_1$, such that $\nextArray[j]> i_1$ (note that such a $j$ always exists, as $\nextArray[i_1]>i_1$). Once $j$ is found, we can safely extend $u$ with $w[j]$, and repeat the procedure, with $i_0$ and $i_1$ updated so that they point to the positions of $w$ where the last two letters of the extended $u$ occur. To obtain the shortest word $v$ which extends $u$ to a $\mas$ our algorithm proceeds in a greedy manner, by choosing $j$ such that $\nextArray[j]$ is the largest in the subarray $\nextArray[i_0+1:i_1]$. We will show that this is correct.

\medskip
So, let us prove that our algorithm works correctly. Firstly, by the remarks we have already made, it is not hard to note that the string $v$ returned by \cref{alg:masExtend} has the property that $uv$ fulfils the requirements from \cref{thm:mas}, so $uv$ is an $\mas$. Now, for the correctness of the greedy part of the algorithm, we need to show that there is no other strictly shorter string $x$ such that $ux$ is an $\mas$.

\medskip
Assume, for the sake of the contradiction, that there exists a string $x$ such that $|x|<|v|$ and $ux$ is an $\mas$. Let us choose $x$ such that $ux$ is an $\mas$, whose length is minimal among all $\mas$ which start with $u$, and, among all such $\mas$ of minimal length, it shares the longest prefix with $uv$.

\medskip
As $|u|>0$ and $u$ itself is not an $\mas$ we have that, $|uv|=h>|ux|=\ell \geq 2$. Let $u_1=uv$ and $u_2=ux$. Let $0=i_0<i_1<\ldots <i_{h-1}$ be positions of $w$ such that $w[i_j]=u_1[j]$ for $j\in [1:h-1]$, as identified greedily by \cref{alg:isSubseq}. Let $0=g_0<g_1<\ldots <g_{\ell-1}$ be positions of $w$ such that $w[g_\ell]=u_2[j]$ for $j\in [1:\ell-1]$, as identified by \cref{alg:isSubseq}. First, let us note that there exists $f\leq \ell-1$ such that $g_f\neq i_f$.  Otherwise, $u_2[1:\ell-1]$ would be a prefix of $u_1$, and we would have a position $j$  in $[g_{\ell-1}:g_\ell]$ such that $\nextArray[j]=\infty$. So our algorithm would have also returned a word of length $\ell$ instead of $v$, a contradiction. Hence, we can choose the smallest $f$ such that $g_f\neq i_f$, and we know that $f\leq \ell-1$. By the way we choose $i_f$ in our algorithm, we have that $g_f<i_f$. Let $d$ be minimum such that $g_{f+d}>i_f$; clearly, $d\geq 1$. We can then replace the sequence of positions where the letters of $x$ appear in $w$ by the sequence $g_1<\ldots <g_{f-1}<i_f<g_{f+d}<\ldots <g_{\ell-1}$ and we note that these positions also fulfil the requirements of \cref{thm:mas}. The word $uw[g_1]\cdots w[g_{f-1}]w[i_f]w[g_{f+d}]...w[g_{\ell-1}] x[\ell]$ is therefore an $\mas$, which is either shorter than $ux$, or has the same length as $ux$ but shares a longer prefix with $uv$. This is a contradiction with the choice of $x$.

\medskip
So, our assumption was false, and it follows that the greedy approach used by \cref{alg:masExtend} produces a word $v$ of minimal length such that $uv$ is an $\mas$. Moreover, the complexity of answering a query is $O(|u|+|v|)$.
\end{proof}

It is worth noting that the lexicographically smallest $\mas$ of a word $w$ is $a^{|w|_a+1}$ where $a$ is the lexicographically smallest letter of $\Sigma$ which occurs in $w$.

\subsection*{Acknowledgements:} This is an extended version of the conference paper \cite{KoscheKMS21}. We thank the anonymous referees, both of the conference and the journal version of this work, for their valuable comments and suggestions. Last but not least, we thank Tina Ringleb and Maximilian Winkler for carefully proof-reading this paper.


\end{document}